\documentclass[12pt]{article}
\usepackage{latexsym,amsmath,amssymb,mathrsfs,graphicx,hyperref,mathpazo,soul,color,comment}
\usepackage{pgf}
\usepackage{graphicx}
\usepackage{comment}
\usepackage[english]{babel}
\usepackage{amsmath}
\usepackage{amsfonts}
\usepackage{relsize}
\usepackage{bm}
\usepackage{arydshln}
\usepackage{mathtools}

\usepackage{hhline}
\usepackage{multirow}

\usepackage{subfig}

\usepackage[round]{natbib}

\usepackage[parfill]{parskip}

\usepackage{authblk}

\usepackage{graphicx,psfrag,epsf}
\newcommand{\blind}{0}

\usepackage{afterpage}

\renewcommand{\ge}{\geqslant}
\renewcommand{\le}{\leqslant}

\def\argmin{\mathop{{\rm argmin}}}
\def\eop{\hfill{$\Box$}\medskip}

\newtheorem{Proposition}{Proposition}

\newcounter{problem}

\newenvironment{proof}{\smallskip\noindent\textbf{Proof.}~}{\phantom{a}\eop}

\date{}

\begin{document}

\def\spacingset#1{\renewcommand{\baselinestretch}%
{#1}\small\normalsize} \spacingset{1}


\if0\blind
{
  \title{\bf Mixture Quantiles Estimated by  Constrained Linear Regression}
  \author{Cheng Peng \\
    Department of Applied Mathematics and Statistics, Stony Brook University\\
    and \\
    Yizhou Li \\
    Department of Applied Mathematics and Statistics, Stony Brook University \\
	and \\
	Stan Uryasev \thanks{Corresponding author. E-mail: stanislav.uryasev@stonybrook.edu}\\
	Department of Applied Mathematics and Statistics, Stony Brook University    
    }
  \maketitle
} \fi

\if1\blind
{
  \bigskip
  \bigskip
  \bigskip
  \begin{center}
    {\LARGE\bf Title}
\end{center}
  \medskip
} \fi

\bigskip
\begin{abstract}
We study the problem of modeling univariate distributions via their quantile functions. We introduce a flexible family of distributions whose quantile function is a linear combination of basis quantiles. Because the model is linear in its parameters, estimation reduces to constrained linear regression, yielding a convex optimization problem that readily accommodates cardinality constraints as well as $L_1$ or smoothness regularization. For $L_q$-type objectives we show the estimator is asymptotically equivalent to a minimum $q$-Wasserstein-distance estimator and establish asymptotic normality. Experiments on simulated and real-world datasets demonstrate that the proposed method accurately captures both the central body and extreme tails of distributions while requiring substantially less computation than standard benchmark approaches.
\end{abstract}

\noindent%
\vfill

\newpage
\spacingset{1.45} 

\section{Introduction}

Building a model for a univariate distribution involves a trade-off among interpretability, flexibility, and tractability.
This paper studies a family of distributions that is sufficiently flexible to accommodate most distributions considered in practice while remaining straightforward to interpret and estimate.

The suggested family of distributions is analogous to mixtures of probability densities constructed as weighted sums of basis density functions (equivalently, a weighted sum of basis CDFs). Instead of aggregating CDFs, we consider a family in which the quantile function (inverse CDF) is defined as a weighted sum of basis quantiles. The idea of modeling a quantile function by adding basic functions appears in \cite{tukey1962future, gilchrist2000statistical}. Various types of basis quantile functions have been considered previously: orthogonal polynomials \citep{Sillitto}; quantiles of the normal and Cauchy distributions mixed with linear and quadratic terms \citep{Karvanen}; and quantiles of modified normal distributions \citep{Keelin2}. The quantile aggregation in \cite{vincent1912vincent} and \cite{ratcliff1979group} has a different motivation: estimating a population distribution by averaging distributions obtained from individuals with equal weights.

Earlier approaches are limited in their choice of basis functions and do not generally guarantee that the resulting quantile function is nondecreasing. Our approach addresses these two issues. The type and number of basis functions in the weighted sum are arbitrary provided each basis function is a valid quantile function. We recommend using quantiles of common distributions (e.g., normal, exponential) when prior knowledge about the shape of the target distribution is available. For more complicated shapes, such as multimodal distributions, monotone I-spline bases are convenient \citep{Ramsay1988Ispline}. Moreover, \cite{Papp2011polynomial, Papp2014spline} show that piecewise I-splines with nonnegative coefficients can approximate any bounded, continuous, nondecreasing function. Both quantiles of known distributions and I-spline bases are continuous and nondecreasing; hence nondecreasingness of the estimated quantile function follows from nonnegativity of the weights (the sum of nondecreasing functions with nonnegative weights is nondecreasing). While Gaussian mixture densities are parametric, the mixture-of-splines approach to estimating the quantile is nonparametric. Nevertheless, the mixture-of-splines model admits a simple closed-form expression and offers substantial flexibility. Monte Carlo simulation is straightforward via inverse transform sampling.

To fit the model, we minimize the distance between the sample quantile function and the model evaluated at sample probabilities. The specific distance depends on the chosen error norm. We formulate the estimation as a constrained linear regression problem, which differs from existing formulations.  Minimum quantile distance estimators \citep{LaRiccia2, Millar, Gilchrist2007quantile} minimize the difference between the sample quantile function and the model, where the distance is expressed as an integral over the unit interval. Related work on minimum Wasserstein (Kantorovich–Rubinstein) distance estimators is surveyed in \cite{Bassetti2, Bassetti, Bernton}. Maximum likelihood estimation is inconvenient for our model because it requires differentiating the inverse of a sum of quantile functions, and the resulting optimization problems are typically nonconvex. In contrast, our linear-regression formulation with linear constraints yields a convex optimization problem that can be solved efficiently with modern convex-programming solvers. Another efficient method, \cite{LaRiccia1}, uses only a fixed number of quantiles. Moment fitting and L-moment fitting \citep{Hosking, Karvanen} require computation of higher-order moments that can be unstable; we show how L-moment constraints can be incorporated into our framework as linear constraints. The quantile regression approach in \cite{Peng2023factor} optimizes a different objective, which is equivalent to the sum of continuous ranked probability scores \citep{Matheson1976score}. When the model contains only one basis function, our estimator reduces to estimating location and scale by weighted least squares \citep{Ogawa, Lloyd}. Using a similar idea, \cite{Kratz1996qqpareto} estimates the tail index of a Pareto distribution using log-transformed data. \cite{alvarez2023quantile} studies mixture quantile models and estimates parameters via a sieve-like generalized method of L-moments, which is related to but distinct from our $L_2$-based error measure.

We employ constraints and regularization in the regression formulation to improve out-of-sample performance on small datasets. A cardinality constraint can limit the maximum number of basis functions selected from a large candidate set; mixed-integer linear programming (MILP) solvers are effective for such problems. We compare the numerical performance of our algorithm with LASSO regression \citep{Santosa1986lasso, Tibshirani1996lasso}.

We study asymptotic properties of the estimator and show that, when the error function is the $L_q$ norm, the estimator is asymptotically a minimum $q$-Wasserstein (Kantorovich–Rubinstein) distance estimator. This result follows by viewing our optimization objective as a Monte Carlo approximation of the integral that defines the difference between two quantile functions. In the nonasymptotic case our objective is not the exact $q$-Wasserstein distance between the empirical quantile function and the model; nevertheless, the objective value converges to the Wasserstein distance asymptotically. This convergence allows the estimator to serve both parameter estimation and goodness-of-fit testing, since the Wasserstein distance has been used for goodness-of-fit tests \citep{Barrio, Panaretos}. We also establish consistency using general results for minimum distance estimators \citep{Newey2}, and we prove asymptotic normality of the estimator (for a fixed number of quantiles) by exploiting properties of order statistics \citep{Rao}. The latter is a special case of the result in \cite{LaRiccia1}.

The distance measured at sample probabilities can also be interpreted geometrically: it equals the sum of horizontal distances from the sample points to the line $y=x$ in a Q–Q plot. Hence the fitted model brings sample points as close as possible to the identity line. Q–Q plots have previously been used for hypothesis testing \citep{Loy2016QQplot}.

Quantile-based models have gained popularity because they provide distinct statistical and computational advantages compared to density-based methods. For a wide range of applications including time series and uncertainty quantification, and recent algorithmic developments, see for example \cite{zhao2025distributed, frizzo2025quantile, vedula2025quantile, deng2025quest, redivo2025mixtures}.

The paper is organized as follows. Section \ref{sec_model} describes the mixture-quantiles model and Section \ref{sec_calibration} defines the regression optimization problem. Section \ref{sec_minimum_distance} proves the relation to minimum Wasserstein estimators. Section \ref{sec_asymptotics_orderstat} establishes the asymptotic normality of the estimator. Section \ref{sec_gls} discusses properties of the estimator obtained by weighted least squares regression. Section \ref{sec_casestudy} contains two numerical experiments. Section \ref{sec_conclusion} concludes the paper.

\section{Mixture Quantiles Model} \label{sec_model}
Let
$I = $ number of basis quantile functions in the mixture quantiles model;
 $p = $ confidence level;
  $\bm{\theta} = ( \theta_0,\cdots,\theta_I )$ = vector of (unknown)  non-negative  coefficients;
  $\{Q_i(p)\}_{i=0}^I$ = basis quantile functions, where $Q_{0}(p) = 1$;
  $G(p,\bm{\theta})=$ quantile model with confidence level $p$ and parameter $\bm{\theta}$;
  $G_{\bm{\theta}}^{-1}(x) = \inf\{p\in (0,1) | G(p,\bm{\theta}) \geq x  \} = $  inverse function of $G(p,\bm{\theta})$ with respect to $p\,$.

This section introduces the mixture quantiles model (linear combination of quantiles) and formulates a linear regression   problem for parameter estimation.  We consider real-valued continuous random variables $X:\Omega \rightarrow \mathbb{R}$ defined on a probability space $(\Omega,\mathcal{F},P)$ with a distribution function $F_X(x) = P(X \leq x)$. The $p$-quantile is defined by $Q(p) = \inf \{ x \in \mathbb{R} | p \leq F_X(x) \}$. Let $\{X_n\}_{n=1}^N$ be $N$ i.i.d. random variables with distribution function  $F_X(x)$. The sample $p$-quantile is the corresponding quantile of the empirical distribution function $\widehat{F}(x) = \frac{1}{N}\sum_{n=1}^N \mathbf{1}_{\{X_i \leq x\}}$.

A location-scale family of distributions is obtained by
scaling and translation
a random
variable with a fixed distribution.  
  Consider an arbitrary quantile function $Q(p)$. The quantile function of the location-scale family  generated by quantile function $Q(p)$ is defined by
\begin{equation}
\mu + \sigma Q(p)  \;,
\end{equation}
where $\mu$ is the location parameter and $\sigma$ is the non-negative scale parameter. 
We construct a new family of quantile functions obtained by linearly combining quantile functions from different location-scale families.
 This new family differs from the mixture densities constructed with these selected distributions.

\paragraph{Model formulation}
 The mixture quantiles model is defined by
\begin{equation} \label{model}
\begin{aligned}
G(p,\bm{\theta})
&=
\sum_{i=0}^I  \theta_i Q_i(p)   
 \;,
\end{aligned}
\end{equation}
where $Q_0(p) = 1$ is the basis function for location parameter.  
The functions  $\{Q_i(p)\}_{i=1}^I$ can be any monotone non-decreasing function defined on the unit interval.  
 $\{Q_i(p)\}_{i=1}^I$ are linearly independent functions, i.e., any function $Q_i(p)$ cannot be represented by a linear combination of $\{Q_j(p)\}_{j\neq i}$. 
 We propose two types of basis functions: quantile functions of common distributions and monotone I-spline basis \citep{Ramsay1988Ispline}.   When  prior knowledge on the shape of the distribution is available, the common quantile functions such as those of normal and exponential distributions are preferred, while the splines are needed for  more complicated case such as multimodal distribution. 
\cite{Papp2011polynomial,Papp2014spline} show that piecewise I-spline with non-negative coefficients can approximate any bounded  continuous   nondecreasing function. We focus on common quantile functions in this study. 

The density corresponding to $G(p,\bm{\theta})$ is given by  $f_{\bm{\theta}}(x) = \frac{\partial}{\partial x} G_{\bm{\theta}}^{-1}(x) = \left( \frac{\partial}{\partial p} G(G_{\bm{\theta}}^{-1}(x),\bm{\theta}) \right)^{-1}$.
While all distributions in one location-scale family have equal skewness, the skewness of the distribution with quantile function $G(p,\bm{\bm{\theta}})$ is dependent on $\bm{\theta}$ nonlinearly if there are more than one quantile function in the mixture. 

We will standardize the quantile functions $\{Q_i(p)\}_{i=1}^I$ to ensure that the columns of the design matrix in the subsequent linear regression problem are on a comparable scale, which improves numerical stability of the convex optimization. Quantile functions  are standardized as follows. For two-tailed distributions, $Q_i(p)$ is standardized such that its 0.5-quantile is zero and unit interquartile range, i.e., $Q_i(0.5) = 0$, $Q_i(3/4)-Q_i(1/4) = 1$. For one-tailed distributions, $Q_i(p)$ is standardized to have zero zero-quantile and unit interquartile range, i.e., $Q_i(0) = 0$, $Q_i(3/4)-Q_i(1/4) = 1$. The I-spline basis functions are already on the same scale and thus do not need to be scaled. We do not standardize by mean and variance since they may not exist for some fat-tailed distributions.

\section{Parameter Estimation by Constrained Linear Regression} \label{sec_calibration}

Let
 $N=$ sample size;
  $J=$ number of selected probabilities; 
 $\{p_j\}_{j=1}^J$ = selected probabilities in ascending order; 
$\{y_j\}_{j=1}^J = $ sample $p_j$-quantiles of  $N$  observations; 
 $\bm{Y}_N = ( y_1, \cdots , y_J  )' = $ vector of $J$ sample quantiles of $N$ observations in ascending order;
 $x_{ji} = Q_{i}(p_j)$;
 $\bm{X} = [x_{ji}]_{J \times (I+1)} =
$ matrix of regression factors.

We want to find a quantile function $G(p,\bm{\theta})$  parametrized by $\bm{\theta}$, which is closest to the sample data. The distance is measured between the sample quantile function and $G(p,\bm{\theta})$ at some selected probabilities $\{p_j\}_{j=1}^J$.  The estimation problem is formulated as a linear regression problem. We can describe the setting in a matrix form
\begin{equation}
\begin{pmatrix}
y_1 \\ y_2 \\ \vdots \\ y_J
\end{pmatrix}
=
\begin{pmatrix}
 1		&	Q_1(p_1)	&	Q_2(p_1) 	& \ldots &Q_I(p_1)	 \\
 1		&	Q_1(p_2)	&	Q_2(p_2)	&		 &	\\
 \vdots	&	\vdots 		&				&\ddots	 &	\\
 1		&	Q_1(p_J)		&				&		 &Q_I(p_J)	\\
\end{pmatrix}_{J \times (I+1)}
\begin{pmatrix}
\theta_0 \\ \theta_1 \\ \vdots \\  \\ \theta_I
\end{pmatrix}
+
\begin{pmatrix}
\epsilon_1 \\ \epsilon_2 \\ \vdots  \\ \epsilon_J
\end{pmatrix}
\;.
\end{equation}
$\{\epsilon_j\}_{j=1}^J$ are the differences between the sample $y_j$ and the $p_j$-quantile of the model. $\{\epsilon_j\}_{j=1}^J$ are not independent or identical.
 Note that  $J$ is not necessarily the total number of observations. When $J=N$, $\bm{Y}$ is the vector of all samples in ascending order.  Sections \ref{sec_minimum_distance} and \ref{sec_asymptotics_orderstat} discuss two cases. The first case uses all available observations as $\bm{Y}$, i.e., $J=N$. The second case uses a fixed number $J$ of the observations as $N \rightarrow \infty$.  

\paragraph{Optimization problem statement}
The distance-minimization problem is
\begin{equation} \label{statement}
\min_{\bm{\theta} \in \bm{\Theta}} \; \bm{\mathcal{E}}(\bm{Y}_N - \bm{X}\bm{\theta}) + \lambda \, \Omega(\bm{\theta}) \;,
\end{equation}
where $\bm{\Theta}$ is a feasible set, $\bm{\mathcal{E}}$ is an error measure, and $\Omega$ is a regularization penalty term. Various error measures and penalty terms can be considered; this paper focuses on weighted $L_q$ norms. To avoid confusion between $p$ (probability) and the $L_p$ norm, we use $q$ to denote the norm index. In particular, we consider the weighted least absolute deviations and weighted least squares regressions corresponding to the weighted $L_1$ and $L_2$ norms, respectively. The considered optimization problem  (\refeq{statement}) is convex with respect to vector parameter $\bm{\theta}$.
Other possible error measures include the regular measures axiomatically defined in \cite{quadrangle}.

\paragraph{Constraints and penalties}
The principal constraint on $\bm{\theta}$ is the nonnegativity constraint
\begin{equation}
\theta_i \ge 0, \quad i = 1,\ldots,I \;,
\end{equation}
which ensures that $G(p,\bm{\theta})$ is nondecreasing in $p$, and thus a valid quantile function. In numerical experiments, components in $\bm{\theta}$ were frequently negative when the nonnegativity constraint was omitted, especially for small sample sizes.

To select a subset of basis quantile functions from a large candidate set, a direct approach is a cardinality constraint. This constraint is defined by the cardinality function
\begin{equation}
\mathrm{Card}(\bm{\theta}) = \sum_{i=1}^I \mathbf{1}_{\{\theta_i \neq 0\}} \le C \;,
\end{equation}
where $\mathbf{1}_{\{\cdot\}}$ denotes the indicator function (equal to 1 when the statement inside the braces is true and 0 otherwise). In the regression setting, this constraint can be handled efficiently by commercial mixed-integer solvers. Selecting basis elements from a mixture-density formulation, by contrast, is often computationally much more expensive (because of the nonconvexity of the likelihood function).

LASSO regression with penalty $\|\bm{\theta}\|_1$ can also be used for selection among basis quantile functions; it usually leads to non-optimal solutions, compared to cardinality constraint but computationally faster due to convexity. LASSO and Ridge regression with penalty $\|\bm{\theta}\|_2^2$ can be used to mitigate overfitting.

For I-spline basis functions, however, sparsity of coefficients may be less interpretable, since individual spline bases do not generally correspond to standalone sensible quantile functions. So small or zero coefficients do not have the same intuitive meaning as when using common-distribution quantiles. In particular, if all spline coefficients except the intercept vanish, the resulting function is constant and thus not a valid quantile function. Inspired by P-splines \citep{eilers2015twenty}, we therefore penalize differences between adjacent coefficients when using I-splines
\begin{align}
\Omega(\bm{\theta}) = \|B\bm{\theta}\|_2^2 \;,
\end{align}
where $B$ is the first-difference matrix so that $B\bm{\theta} = (\theta_2-\theta_1,\ldots,\theta_I-\theta_{I-1})'$. In the limit of a very large penalty, the model reduces to a straight line, which corresponds to the quantile function of a uniform distribution. Appendix \ref{sec_optimization} further discusses constraints, penalties, and error measures in the optimization problem.

\paragraph{Estimation of probabilities}
Consider the case $J=N$. Observe that if the selected probabilities $\{p_n\}_{n=1}^N$ equal the true probabilities $\{G_{\bm{\theta}^*}^{-1}(y_n)\}_{n=1}^N$ for the true parameter $\bm{\theta}^*$, then $\bm{\theta}^*$ is an optimal solution for \eqref{statement} (for any choice of $\bm{\mathcal{E}}$), because the residuals are zero. Since the true probabilities are unknown, we adopt a standard  estimator
\begin{equation}
p_n = \frac{n}{N+1}, \quad n = 1,\ldots,N \;,
\end{equation}
although other conventions for empirical quantiles exist. Note that the $i$-th order statistic $y_i$ is an asymptotically unbiased estimator of the $p_i$-quantile of the distribution with quantile function $G(p,\bm{\theta})$. There is no universal consensus on the best definition of sample quantiles \citep{Hosseini2016position}.

\paragraph{Q-Q Plot}
In a Q-Q plot, the horizontal distance between a data point $(G(p_n,\bm{\theta}),y_n)$ and the straight line $y=x$ is $|Q(p_i) - G(p,\bm{\theta})|^q$. Thus the estimation problem (\ref{statement}) can be viewed as  obtaining the best Q-Q plot by  minimizing the average distance between the data points $\{(G(p_n,\bm{\theta}),y_n)\}_{n=1}^N$ and the straight line. It is known from empirical evidence that  the Q-Q plot of samples against true quantiles often exhibits a noticeable deviation from the straight line at both ends, where extreme $y_n$ values reside. This finding provides additional support for the proposed downweighting approach suggested in eariler discussion of the sample probabilities.

\section{Minimum Wasserstein Distance Estimator}\label{sec_minimum_distance}

We make the following notations in this section:
 $\{w_j\}_{j=1}^J =$ non-negative weights;
 $w(p)=$ non-negative integrable weight function;
 $\mathcal{E}_q(x) = |x|^q$ ;  
 $\bm{\Theta}=$  feasible  set; 
 $\bm{x}_j=$  $j$-th row of matrix $\bm{X}$ with dimension ${J \times (I+1)}$;
 $$f_N(\bm{\theta})  = \left( \frac{1}{J} \sum_{j=1}^J w_j \mathcal{E}_q(y_j - \bm{x}_j \bm{\theta})  \right)^{\frac{1}{q}} =\text{  weighted variant of objective function in \eqref{statement}}; $$  
 $$f(\bm{\theta})  = \lim_{J\rightarrow\infty} \left( \frac{1}{J} \sum_{j=1}^J w_j \mathcal{E}_q(y_j - \bm{x}_j \bm{\theta}) \right)^{\frac{1}{q}}; $$
$\widehat{\bm{\theta}}_N 
=
 (\widehat{\theta}_{N0},\widehat{\theta}_{N1},\cdots,\widehat{\theta}_{NI})' 
 \in
  \argmin_{\bm{\theta}\in\bm{\Theta}} f_N(\bm{\theta}) =$  estimator with sample size $N$; 
 $\widehat{\bm{\theta}}_0 = \argmin_{\bm{\theta}\in\bm{\Theta}} f(\bm{\theta})$;
 $\bm{\theta}^{*} = (\theta_{0}^*,\theta_{1}^*,\cdots,\theta_{I}^*)' =$ true parameter vector; 
 The weighted Wasserstein distance between two distributions with quantile functions $Q_1(p)$ and $Q_2(p)$ is defined by
\begin{align}\label{def_wasserstein}
\mathcal{W}_q(Q_1(p),Q_2(p)) = \left( 
\int_0^1 \mathcal{E}_q(Q_1(p) - Q_2(p)) w(p) \mathrm{d}p 
 \right)^{\frac{1}{q}}.
 \end{align}

  This section considers the case where the number of selected probabilities $J$ is equal to the sample size $N$.  We have $p_n = \frac{n}{N+1}$, $n=1,\cdots,N$. We prove that the estimator is asymptotically the minimum Wasserstein distance estimator. 
 The model is said to be { correctly specified} if the true quantile function is in the form of \eqref{model} and the true parameter $\bm{\theta}^*$ is in the interior of the feasible set $\bm{\Theta}$.

\bigskip

\begin{Proposition}\label{prop_consistent}
If  (i) $\widehat{\bm{\theta}}_0$ is a unique minimizer of  $f(\bm{\theta})$ ;  (ii) $\widehat{\bm{\theta}}_0$ is
an element in the interior of a convex set $\bm{\Theta}$, then $\widehat{\bm{\theta}}_N \overset{p}{\rightarrow} \widehat{\bm{\theta}}_0$. If the model is correctly specified,  then $\widehat{\bm{\theta}}_N$ is a consistent estimator.  
\end{Proposition}

 Since  $\{Q_i(p)\}_{i=0}^I$ are linearly independent, $G(p,\bm{\theta}_1) = G(p,\bm{\theta}_2)$ if and only if $\bm{\theta}_1 = \bm{\theta}_2$ for $\bm{\theta}_1,\bm{\theta}_2 \in \bm{\Theta}$. Thus if the model is correctly specified, there  exists a unique minimizer $\bm{\theta}^* \in \bm{\Theta}$, i.e.,  the  assumption {\it (i)}   is automatically satisfied. 
 Proposition \ref{prop_consistent} is a general result that holds for convex error functions ${\mathcal{E}}$.

\begin{Proposition}\label{prop_integral_wass}
For $\bm{\theta} \in \bm{\Theta}$, $f_N(\bm{\theta})$ converges in probability to the weighted $q$-Wasserstein distance  between the model $G(p,\bm{\theta})$ and true quantile function $Q(p)$.
\end{Proposition}

The proofs are deferred to Appendix \ref{sec_proof_prop_consistent} and \ref{sec_proof_prop_integral_wass}. Note that although the estimator converges to the one that minimizes the $q$-Wasserstein distance between the model and true quantile function, the objective function $f_N(\bm{\theta})$ is not the $q$-Wasserstein distance.
 Furthermore, with the linear formulation with respect to $\bm{\theta}$, we can bound the $q$-Wasserstein distance $\mathcal{W}_q(Q_1(p),Q_2(p))$ by the scaled $L_1$ distance between the estimator and true parameters.

\begin{Proposition} \label{prop_wass}
Suppose (i) the conditions in Proposition \ref{prop_consistent} are satisfied; (ii) the model is correctly specified; (iii) $\forall i, (\int_0^1 \mathcal{E}_q\left(    Q_i(p)  \right) w(p) \mathrm{d}p)^{1/q} \leq M$, where $M$ is a constant.  Then the $q$-Wasserstein distance ${\mathcal{W}}$ between the model and true distribution is bounded  by the scaled $L_1$ distance between the estimator and true parameters.
\end{Proposition}

The proof is deferred to Appendix \ref{sec_proof_prop_wass}. Proposition \ref{prop_wass}  holds for any convex and positively homogeneous
 error function. For a random variable with quantile function $Q(p)$, the lower and upper partial $q$-moments with a reference confidence level $p_0 \in(0,1)$  are defined by $\int_{-\infty}^{p_0} Q(p)^q\mathrm{d}p$ and $\int_{p_0}^{\infty} Q(p)^q \mathrm{d}p$, respectively (for definition and  application  in finance, see
  \cite{BAWA1977partialmoment_CAPM}). When $w(p)=1$, the condition $(iii)$ is equivalent to the requirement that the random variables with quantile function $\{Q_i(p)\}_{i=1}^I$ have finite lower and upper partial moments.  We provide a simulation study in Appendix \ref{sec_sim_convergence} to verify the convergence.

\section{Asymptotic Normality}\label{sec_asymptotics_orderstat}
Let
 $G_{\bm{\theta}}^{-1}(y)=$ inverse function of $G(p,\bm{\theta})$ with respect to $p$;
 $\bm{W}=$ weight matrix;
  $\bm{C} = [c_{ij}]_{J \times J}
=$ asymptotic covariance matrix of difference between sample quantile and true quantile,  $ c_{ij} = \frac{p_i(1-p_j)}{ f(G_{\bm{\theta}}^{-1}(p_i)) f(G_{\bm{\theta}}^{-1}(p_j))} ,\quad i \leq j \quad \text{and} \;\;
c_{ij} = c_{ji} ,\quad i \leq j $ ;
 $\bm{H}
=
(\bm{X}'\bm{W}\bm{X})^{-1} \bm{X}' \bm{W} \bm{C} \bm{W}' \bm{X} (\bm{X}' \bm{W}' \bm{X})^{-1} $ ;
 $\bm{H}_0 = 
 (\bm{X}'\bm{C}^{-1}\bm{X})^{-1}$ ;
 $\bm{Y}^* = (G_{\bm{\theta}^*}^{-1}(p_1),G_{\bm{\theta}^*}^{-1}(p_2),\cdots,G_{\bm{\theta}^*}^{-1}(p_J))=$ vector of true quantiles of probabilities $\{p_j\}_{j=1}^J$ .

We study an estimator that selects a fixed number $J$ of probabilities in this section. This is in contrast to Section \ref{sec_minimum_distance} where $J=N$. The estimator is more robust to outliers since only the sample $p_i$-quantiles  are used in the regression. Here we consider the weighted $L_2$ norm as the error measure $\bm{\mathcal{E}}$ in \eqref{statement}.  The  proof of the following proposition  is similar to \cite{Ogawa}  but presented in matrix form.

\begin{Proposition}\label{prop_orderstat} If  
  (i) $G_{\bm{\theta}}^{-1}(y)$ is differentiable w.r.t.  $y$;  (ii) the density function $\frac{\partial}{\partial y}G_{\bm{\theta}}^{-1}(y)$ is positive and bounded; (iii)  $\bm{X}'\bm{W}\bm{X}$ is invertible, then
  \medskip
\begin{itemize}
\item The estimator converges in distribution to the normal distribution
\begin{equation}
\sqrt{N}
(
\widehat{\bm{\theta}}_N - \bm{\theta}^{*}
)
 \overset{d}{\rightarrow}
  \mathcal{N}(\bm{0},\bm{H}) \;.
\end{equation}
\item The optimal weight for the best linear unbiased estimator (BLUE) that minimizes the variance  is
\begin{equation} \label{opt_weight}
\bm{W} = \bm{C}^{-1} \;.
\end{equation}
\item The optimal estimator obtained with the optimal weight converges in distribution to a normal distribution
\begin{equation} \label{opt_estimator}
\sqrt{N}
(
\widehat{\bm{\theta}}_N - \bm{\theta}^{*}
)
 \overset{d}{\rightarrow} 
 \mathcal{N}(\bm{0} , \bm{H}_0) \;,
\end{equation} 
where $\bm{H}_0 = 
 (\bm{X}'\bm{C}^{-1}\bm{X})^{-1}$.
\end{itemize}
\end{Proposition}

The proof is deferred to Appendix \ref{sec_proof_prop_orderstat}. \cite{LaRiccia1} studies a general case where $G(p,\bm{\theta})$ is not linear with respect to $\bm{\theta}$. The asymptotic normality is proved directly with the inverse function theorem and Taylor series expansion. The resultant optimal weight matrix $\bm{W}$ is the same as what we obtain in \eqref{optimal_weight}. 

The weight matrix obtained in \eqref{opt_weight} is not directly applicable, since it requires knowing the covariance matrix, which is a function of the unknown density. The feasible weighted least squares regression uses an estimated covariance matrix to replace matrix $\bm{H}$ in the regression. It can be obtained by calculating the covariance matrix of the residuals of ordinary least squares regression, or by plugging point density estimation in matrix $\bm{C}$. 
Though these estimators may result in suboptimal weight, empirical results in \cite{Ergashev} show that the impact is insignificant  as long as the extreme tail observations are assigned with small weights.

\section{Discussion of Estimation with Weighted Least Squares Regression} \label{sec_gls}

This section discusses some properties of the estimator obtained by weighted least squares regression. Since the sample size $N$ does not play a role in the conclusion of this section, we omit the superscripts and subscripts $N$. 
The objective function equals
\begin{equation}
\bm{\mathcal{E}} ( \bm{Y} - \bm{X} \bm{\theta} )  =  ( \bm{Y} - \bm{X} \bm{\theta} )' \bm{W}  ( \bm{Y} - \bm{X} \bm{\theta} ) \;,
\end{equation}
where $\bm{W} = [w_{ij}]_{J \times J}$ is a symmetric  positive-definite weight matrix. 
The estimator has a  closed form solution
\begin{equation} \label{closedform}
\widehat{\bm{\theta}}  = (\bm{X}'\bm{W}\bm{X})^{-1}\bm{X}'\bm{W}\bm{Y} \;.
\end{equation}
Notice that $\widehat{\theta}_i$  are linear combinations of the observations $y_i$. Thus, it can be regarded as an L-estimator \citep{david2004order}. 
The results in this section all hold for a symmetric matrix $\bm{W}$.

\subsection{Equivariance}
Let 
$\widehat{\bm{\theta}} = (\widehat{\theta}_0, \widehat{\theta}_1,\cdots,\widehat{\theta}_I)$  $=$ the estimator;
 $\widetilde{\bm{\theta}}$ $=$ the estimator with scaled standard distributions $Q_i(p)$ and scaled and shifted $\bm{Y}$;
  $\bm{D}$ $=$   $\text{diag}(1,k_1,k_2,\cdots,k_I)$ = diagonal matrix of scale parameters of $Q_i(p)$;
  $\bm{S}$ $=$ $(s,s,\cdots,s)'$ = location shift vector of data $\bm{Y}$;
 $\bm{S}_0 = (s,0,0,\cdots,0)'$ ;
 $\sigma$ $=$ scale parameter of data $\bm{Y}$ .

The estimator obtained by weighted least squares regression possesses  nice properties   such as location-scale equivariance \citep{lehmann2006theory}. In addition, we can study the impact of nonnormalized basis quantile function. If we replace a quantile function $Q_i(p)$ in the mixture with a scaled one $k Q_i(p)$, the scale estimator $\widetilde{\theta}_i$ becomes $\frac{1}{k} \widetilde{\theta}_i$  while the other scale parameters and the location parameter $\{\widetilde{\theta}_j\}_{j \neq i}$  do not vary. 
These properties can be obtained by the following direct matrix calculation using the closed form solution \eqref{closedform}
\begin{equation}\label{GLS_property}
\widetilde{\bm{\theta}} = (\bm{D} \bm{X}'\bm{W}\bm{X} \bm{D})^{-1} \bm{D} \bm{X}'\bm{W}(\sigma \bm{Y} + \bm{S}) = \sigma \bm{D}^{-1} \widehat{\bm{\theta}} + \bm{S}_0	\;.
\end{equation}

\subsection{Model with Single Basis Function}

We consider the simplest case where there is only one component $Q(p)$ in the mixture and the error function is $L_2$ norm. We estimate the location and scale parameters $\theta_0$, $\theta_1$ in a location-scale family with ordinary least squares regression. Using identity matrix as $\bm{W}$ in \eqref{closedform}, we obtain
\begin{equation} \label{univariate_sol}
 \widehat{\theta}_0 
 = \sum_{n=1}^N \frac{(\sum_{n=1}^N q_n^2) - (\sum_{n=1}^N q_n)q_n}{N(\sum_{n=1}^N q_n^2) - (\sum_{n=1}^N q_n)^2}  y_n \;,\quad
 \widehat{\theta}_1 = \sum_{n=1}^N \frac{  -(\sum_{n=1}^N q_n) + N q_n }{N(\sum_{n=1}^N q_n^2) - (\sum_{n=1}^N q_n)^2}  y_n \;,
\end{equation}
where $q_n = Q(p_n)$, $n=1,\ldots , N$.

Both the location and the scale estimator are a weighted sum of the observations $\{y_n\}_{n=1}^N$. If the true distribution is symmetric, we have $\sum_{n=1}^N q_n = 0$. We have that $\widehat{\theta}_0$ is the sample mean $\sum_{n=1}^N y_n / N$.
If the observations $y_n$ equal exactly the $\frac{n}{N+1}$-quantiles of the true distribution, i.e., $y_n = \theta_0 + \theta_1 q_n$, we recover the true parameters $\widehat{\theta}_0 = \theta_0$, $\widehat{\theta}_1 = \theta_1$. 

\section{Case Study}\label{sec_casestudy}

This section contains numerical experiments on simulated and real-world data to demonstrate the statistical and computational advantage of our approach. Section \ref{sec_gmm} compares with a well-specified Gaussian mixture model in simulation. Section \ref{sec_electricity} compares with Gaussian mixture model and Pareto distribution on body and tail fitness on electricity price dataset. Section \ref{sec_sim_convergence}  studies the convergence of Wasserstein distance bewteen the mixture quantile model and the true quantile function with a simulated dataset. Section \ref{sec_dd} compares the goodness-of-fit of basis quantile function selected via LASSO regression and a cardinality constraint. For implementation details, see the GitHub repository \href{https://github.com/yzli90/Mixture_Quantile}{https://github.com/yzli90/Mixture\_Quantile}.

\subsection{Comparison with Well-Specified Gaussian Mixture Model in Simulation}\label{sec_gmm}

We simulate data from a Gaussian mixture distribution and compare the fitted mixture quantile (MQ) model against a well-specified Gaussian mixture model (GMM). The aim is to evaluate whether the mixture quantile model can attain comparable statistical performance and likelihood despite the GMM being the true model.

\paragraph{Data} We simulate $N=5{,}000$ independent observations from the following five-component Gaussian mixture
$$\frac{3}{10}\mathcal{N}(10, 0.04) + \frac{3}{10}\mathcal{N}(10, 6.25) + \frac{1}{10}\mathcal{N}(12, 4) + \frac{2}{10}\mathcal{N}(13, 4) + \frac{1}{10}\mathcal{N}(14, 4)$$
where $\mathcal{N}(\mu,\sigma^2)$ denotes the normal distribution with mean $\mu$ and variance $\sigma^2$.

\paragraph{Experimental setup}
We perform 5-fold cross validation on the simulated dataset. The GMM benchmark is a 5-component Gaussian mixture fitted by the SciPy package. For the mixture quantile model, we use an I-spline basis with 100 knots and order 3, plus an additional Gaussian quantile function as the basis quantile functions. The P-spline smoothness penalty is fixed at 100  and not tuned.  The parameter estimation problem of MQ is solved by GUROBI.  For comparable optimization precision, the tolerance parameter for the solver in the SciPy package is set to $10^{-6}$, roughly matching GUROBI's tolerance. The maximum iteration limit is set to $1,000$.

\paragraph{Results}
Table \ref{table_gmm} reports averaged metrics across cross-validation folds. The principal observations are:
\begin{itemize}
\item \textbf{Computation time:} MQ achieves markedly faster training time than GMM for the reported in-sample fit.
\item \textbf{In-sample fit:} MQ attains smaller MSE and MAE (the MQ objective is $L_2$/quantile-based), while GMM attains a slightly better KS statistic.
\item \textbf{Out-of-sample fit:} Both models perform comparably; out-of-sample MSEs are essentially identical and MAE differs marginally. GMM attains a slightly better out-of-sample log-likelihood, as expected given that MLE optimizes likelihood for the true-model family. Notably, MQ achieves a log-likelihood very close to that of the true-model GMM, indicating that the quantile-mixture formulation can approximate the DGP well even when the DGP is a Gaussian mixture.
\end{itemize}

Overall, the out-of-sample metrics are sufficiently close to conclude that MQ and GMM exhibit comparable statistical performance in this well-specified setting. Figure \ref{fig_gmm} displays the in-sample density and quantile plots of the results from the last fold in cross validation. the MQ density shows small wiggles, suggesting that increased smoothing could be beneficial. Both models align closely with the empirical quantile function.

\begin{table}[htbp]
\centering
\begin{tabular}{|l|c|c|c|c|}
\hline
 & \multicolumn{2}{c|}{In-sample} & \multicolumn{2}{c|}{Out-of-sample} \\
\cline{2-5}
Metric & MQ  & GMM & MQ  & GMM \\
\hline
Training time (s) & $\mathbf{0.0810}$ & $0.2712$ &  $-$ & $-$  \\
\hline
MAE (Wasserstein) & $\mathbf{0.0094}$ & $0.0175$ & $0.0908$ & $\mathbf{0.0905}$ \\
\hline
MSE & $\mathbf{0.0006}$ & $0.0028$ & $\mathbf{0.0205}$ & $0.0205$ \\
\hline
KS distance & $0.0096$ & $\mathbf{0.0058}$ & $0.0289$ & $\mathbf{0.0285}$ \\
\hline
Log likelihood & $\mathbf{-2.0595}$ & $-2.0612$ & $-2.0727$ & $\mathbf{-2.0630}$ \\
\hline
\end{tabular}
\caption{Comparison of in-sample and out-of-sample performance on simulated data from GMM. MQ represents mixture quantile model. GMM represents Gaussian mixture model. All metrics are average values from folds in cross validation. MAE and MSE are distances between empirical quantile function and model quantile function.}
\label{table_gmm}
\end{table}

\begin{figure}[htbp]
    \centering
    \includegraphics[width=0.48\linewidth]{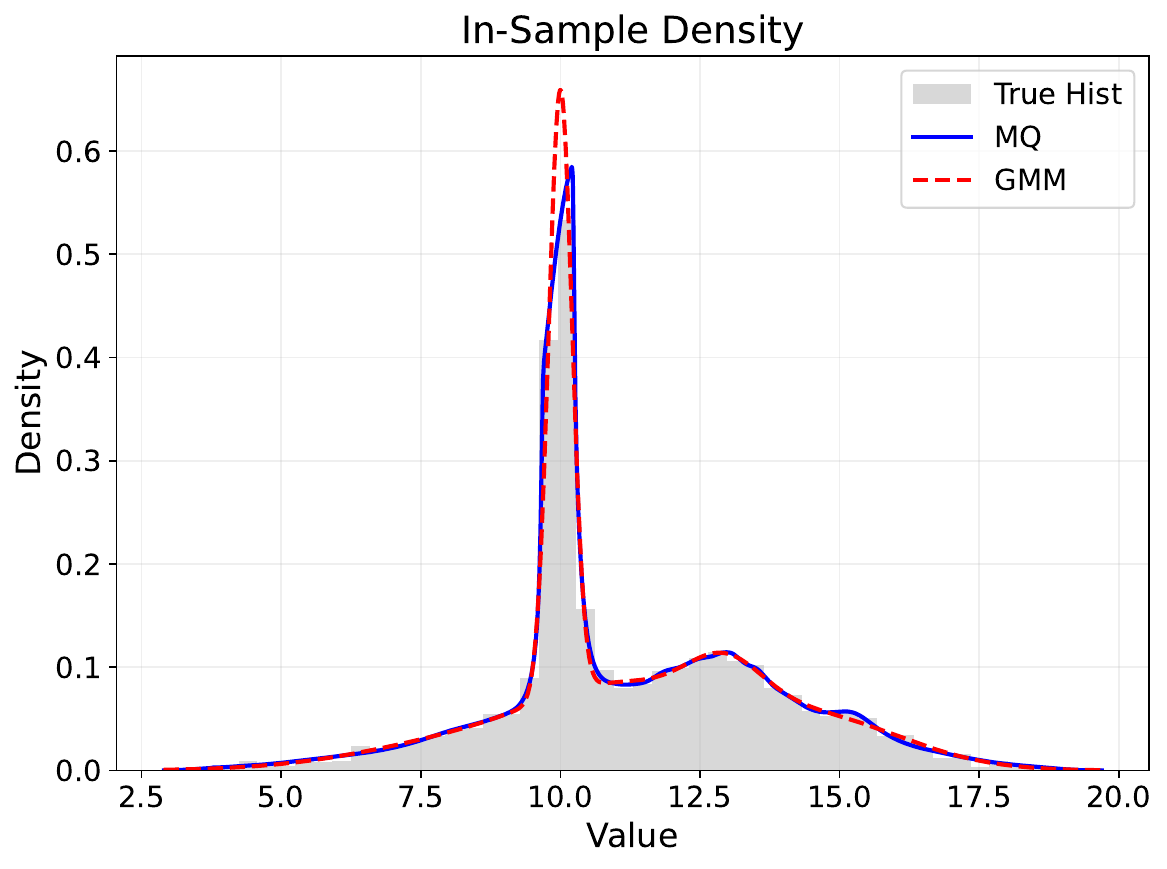}
    \hspace{3mm}
    \includegraphics[width=0.36\linewidth]{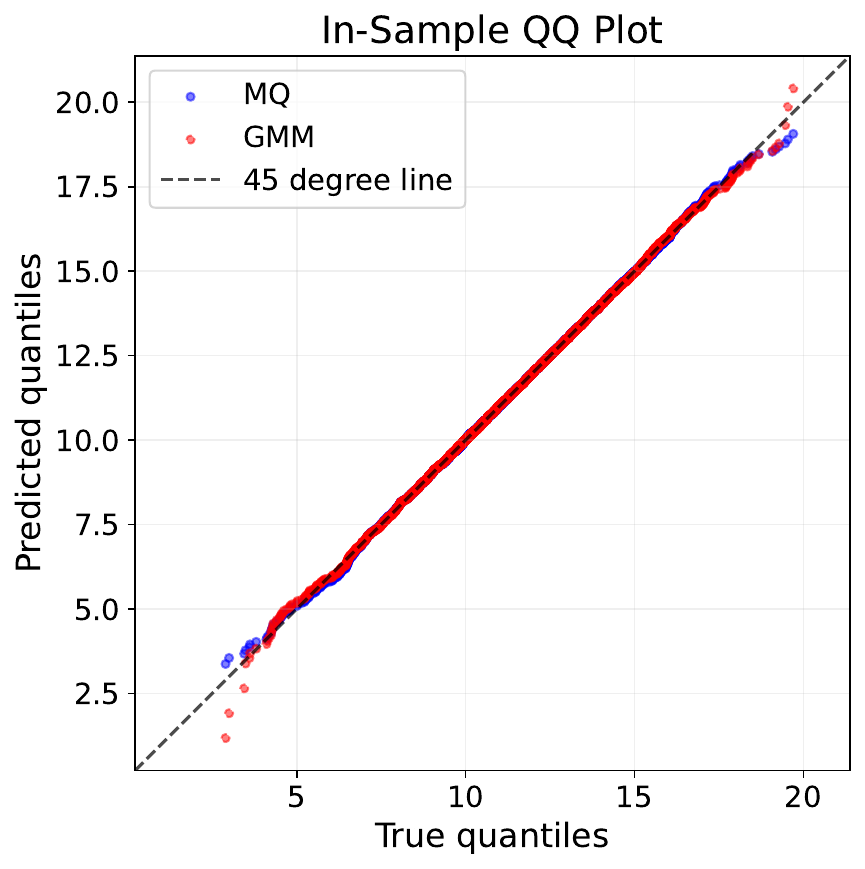}
    \caption{Comparison of density function and quantile function and emperical data in the last fold of cross validation. }
    \label{fig_gmm}
\end{figure}

\subsection{Body and Tail Fit on Electricity Price Dataset}\label{sec_electricity}
This section evaluates the ability of the proposed method to fit both the central body and the tail of real-world electricity price data, which exhibit a skewed body and a long right tail. In addition to a 10-component Gaussian mixture model as a benchmark, we also consider a conventional peak-over-threshold approach that fits a generalized Pareto distribution (GPD) to $5\%$ tail observations. Tail-specific metrics are used to compare the proposed mixture quantile model (MQ) with the GPD-based approach.

\paragraph{Data}
The dataset consists of daily U.S.\ electricity prices from 2001-01-01 to 2024-01-01, sourced from the U.S.\ Energy Information Administration \citep{king_electricity_prices_2024}.

\paragraph{Experimental setup}
We apply 5-fold cross-validation. The basis quantile functions for the mixture quantile model include:
\begin{itemize}
\item Standard Gaussian distribution
\item Student's $t$ distribution with degrees of freedom $1,5,10,30$
\item Generalized Pareto distribution with shape parameters $\xi = -1,-0.5,0,0.5,1$, each scaled and shifted to $[0.95,1]$ to model the upper tail
\item I-spline bases with 5 knots and order 3.
\end{itemize}
The $L_1$ penalty is applied only to basis quantile functions of known distributions, namely, Standard Gaussian, Student's t and generalized Peratio. The $L_1$ penalty is set to 100. The I-spline smoothness penalty is set to 100. The GMM benchmark is a 10-component Gaussian mixture.

\paragraph{Results}
Table \ref{table_electricity} summarizes averaged metrics. Key observations are as follows.
\begin{itemize}
\item \textbf{Computation time:} MQ is extremely efficient relative to the GMM, reflecting the convexity and small basis set of the regression formulation.
\item \textbf{Global fit:} GMM attains smaller MAE and higher log-likelihood, indicating a tighter fit in the central body of the distribution. However, MQ obtains substantially smaller MSE and markedly better tail metrics (0.95 and 0.99 quantile MAE/MSE and expected shortfall differences), demonstrating superior tail fit.
\item \textbf{Tail behavior:} The GMM fails to capture heavy tails, reflected by much larger tail MSEs and ES diffs. By contrast, the rich tail structure of MQ allows for capturing extreme quantiles considerably better, with out-of-sample improvements that are practically relevant for risk assessment. 
\end{itemize}

The Zipf plot in Figure \ref{fig_electricity} suggests that the GMM may have misallocated a component due to overfit, while MQ provides a more coherent fit across orders of magnitude. The density plot  in Figure \ref{fig_electricity}  shows both models tracking the main mass of the data. MQ exhibits a sharper drop in density near the tail threshold owing to the piecewise inclusion of the Pareto tail starting at $p=0.95$ (a kink in the quantile function).

\begin{table}[htbp]
\centering
\begin{tabular}{|l|c|c|c|c|c|c|}
\hline
 & \multicolumn{3}{c|}{In-sample} & \multicolumn{3}{c|}{Out-of-sample} \\
\cline{2-7}
Metric & MQ & GMM & GPD & MQ & GMM & GPD \\
\hline
Training time (s) & $\mathbf{0.0017}$ & $13.0884$ & $0.0425$ & $-$ & $-$ & $-$ \\
\hline
MAE (Wasserstein) & $\mathbf{0.0367}$ & $0.0394$ & $-$ & $\mathbf{0.0739}$ & $0.0874$ & $-$ \\
\hline
MSE & $\mathbf{0.0207}$ & $0.9220$ & $-$ & $\mathbf{0.1059}$ & $0.9238$ & $-$ \\
\hline
KS distance & $0.0054$ & $\mathbf{0.0027}$ & $-$ & $0.0094$ & $\mathbf{0.0082}$ & $-$ \\
\hline
Log likelihood & $-2.6777$ & $\mathbf{-2.6749}$ & $-$ & $-2.6783$ & $\mathbf{-2.6752}$ & $-$ \\
\hline
0.95-quantile MAE  & $\mathbf{0.3367}$ & $0.5889$ & $0.3662$ & $0.5680$ & $0.9142$ & $\mathbf{0.5427}$ \\
\hline
0.95-quantile MSE  & $\mathbf{0.3725}$ & $18.3778$ & $1.4345$ & $\mathbf{2.0047}$ & $18.3510$ & $2.2074$ \\
\hline
0.95-quantile diff & $0.4587$ & $0.0236$ & $\mathbf{0.0099}$ & $0.4606$ & $\mathbf{0.1535}$ & $0.1704$ \\
\hline
0.95-ES diff & $0.0609$ & $0.5337$ & $\mathbf{0.0097}$ & $\mathbf{0.3400}$ & $0.5753$ & $0.3403$ \\
\hline
0.99-quantile MAE  & $\mathbf{0.5999}$ & $2.7711$ & $0.9541$ & $\mathbf{0.9228}$ & $2.9905$ & $1.0437$ \\
\hline
0.99-quantile MSE  & $\mathbf{1.4556}$ & $91.8923$ & $6.8532$ & $\mathbf{8.6765}$ & $90.7654$ & $9.7824$ \\
\hline
0.99-quantile diff  & $\mathbf{0.0549}$ & $0.2886$ & $0.4588$ & $\mathbf{0.4693}$ & $0.7428$ & $0.5612$ \\
\hline
0.99-ES diff  & $0.3180$ & $2.6366$ & $\mathbf{0.0739}$ & $0.7081$ & $2.4927$ & $\mathbf{0.5237}$ \\
\hline
\end{tabular}
\caption{Comparison of in-sample and out-of-sample performance on U.S. electricity price data. MQ denotes the mixture quantile model; GMM denotes the Gaussian mixture model. Metrics are averaged across cross-validation folds. $\alpha$-quantile MAE/MSE refers to distances between empirical and model quantile functions restricted to $[ \alpha, 1]$. $\alpha$-quantile diff is the absolute difference between model and empirical $\alpha$-quantile. ES denotes expected shortfall. The generalized Pareto distribution (GPD) is fitted with $5\%$ tail observations, hence does not have global metrics.}
\label{table_electricity}
\end{table}

\begin{figure}[htbp]
    \centering
    \includegraphics[width=0.48\linewidth]{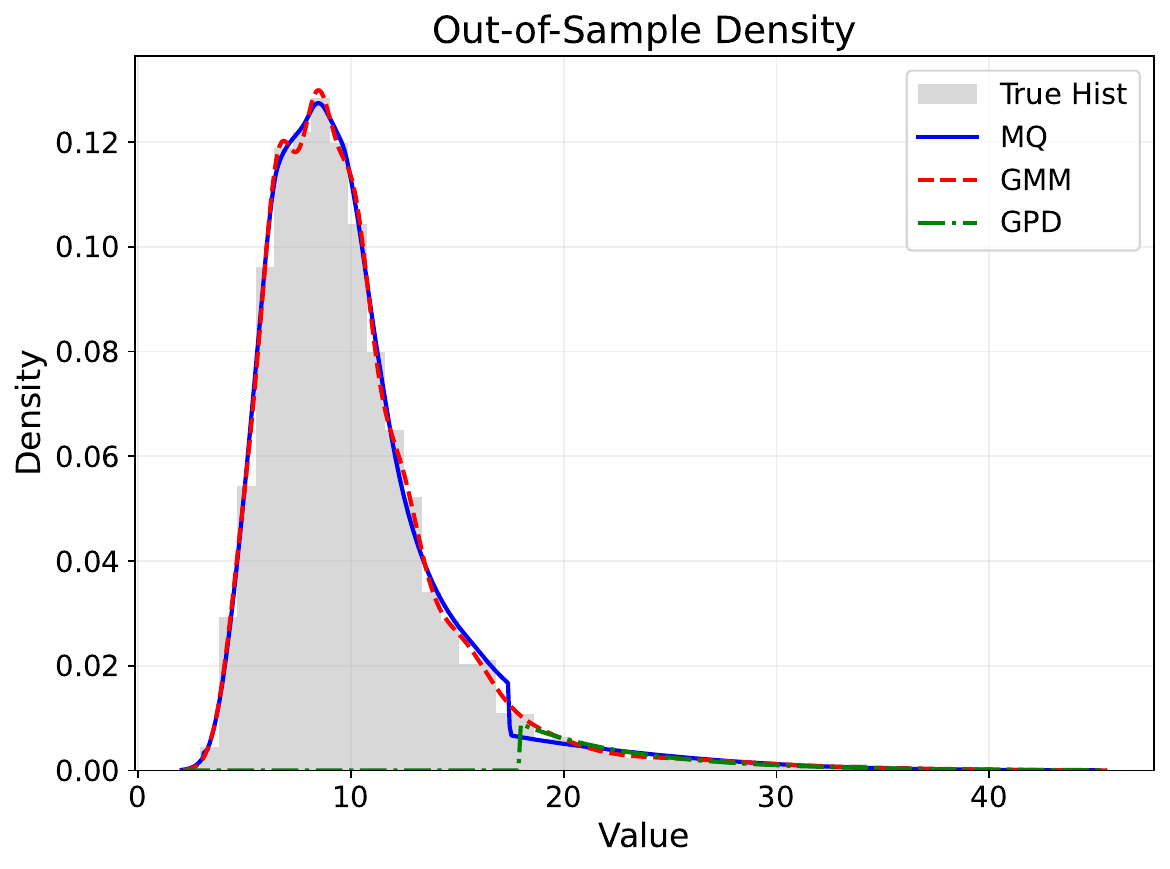}
    \hfill
        \includegraphics[width=0.48\linewidth]{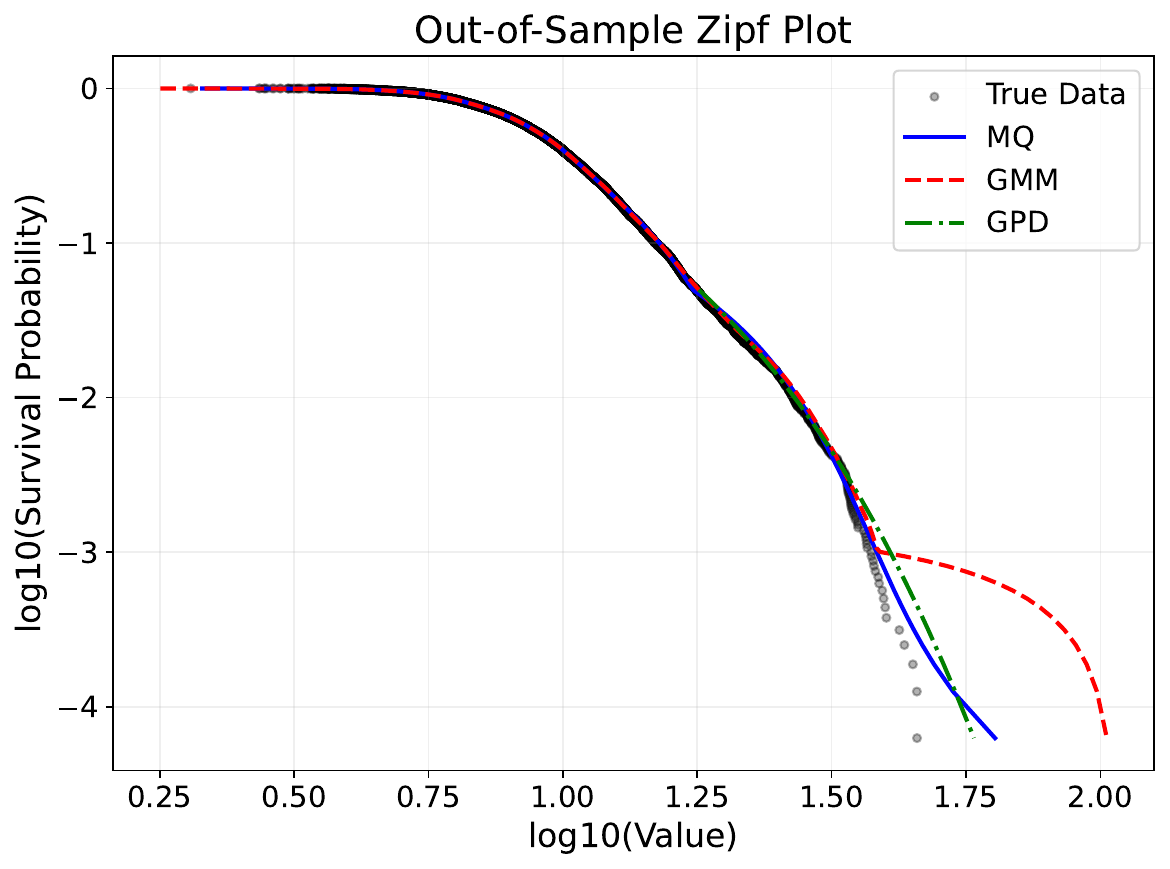}
    \caption{Comparison of density function and Zipf plot in the last fold of cross validation. }
    \label{fig_electricity}
\end{figure}

\subsection{Estimation of Mixture Quantiles of Generalized Beta Distributions with Financial Market Drawdowns}
\label{sec_dd}

This section compares the goodness-of-fit of basis quantile function selected via LASSO regression and a cardinality constraint. We fit the mixture quantiles of generalized beta distributions of the second kind to the log market drawdowns. This case study is inspired by \cite{Frey} which finds that the  lomax distribution has a good fit for market drawdowns.  Drawdowns are calculated with  40-year historical adjusted close price of the SP500 Index from 01/01/1982 to 01/01/2022. 

\paragraph{Financial Market Drawdowns}
Drawdown  financial characteristics is the  the security price drop from the previous highest value in percentage. Suppose we have a time series of security  price,  $p_t$, $t=1,\cdots,T$. The drawdown at time $t$ is defined by $d_t =  (\max_{i\in \{1,\cdots,t \}} p_i - p_t) / \max_{i\in \{1,\cdots,t \}} p_i$. The security is in a drawdown at time $t$, if $p_t > 0$. We call period $(t_1,t_2)$ a drawdown period  if ${d}_t > 0$ for $t_1<t<t_2$ and ${d}_{t_1} = {d}_{t_2} = 0$.   We take the maximum drawdown in each drawdown period.
SP500 Index has a total of $406$ drawdown periods from 01/01/1982 to 01/01/2022.  Then, we fit a model to 406 maximum log-tranformed drawdowns. 

\paragraph{Generalized Beta Distribution of the Second Kind}
The generalized beta distribution of the second kind is defined by the following density function $f(x,\theta_1,\theta_2,\theta_3,\theta_4) = \frac{|\theta_1|y^{\theta_1 \theta_3-1}}{b^{\theta_1 \theta_3}B(\theta_3,\theta_4)(1+(x/\theta_2)^{\theta_1})^{\theta_3+\theta_4}}$. It nests many common distributions such as the generalized gamma (GG), Burr type 3, lognormal, Weibull, Lomax, half-Student's $t$, exponential, power function, etc. Our approach can be regarded as estimating  mixture quantiles of these distributions. We first use the following grid to populate the design matrix. $\theta_2 \in \{0.1,0.2,\cdots,1\}$, $\theta_3 \in \{1,2,\cdots,10\}$, $\theta_4 \in \{1,2,\cdots,10\}$. The scale parameter $\theta_1$ is fixed at $1$. The quantile functions are then standardized as described in Section \ref{sec_calibration}. Then, to improve stability of  optimization, we remove extreme values:  quantile functions that either have $1/407$-quantile smaller than $10^{-4}$ or $406/407$-quantile larger than $10^{3}$. This is a reasonable choice since the scale of drawdowns varies from $10^{-2}$ to $10^2$. Those removed components is unlikely to be selected by the algorithm as a good fit, anyway. A total of $4792$ quantile functions are included in the initial mixture.

\paragraph{Optimization Setting}  
We minimize the error (see \eqref{statement}), 
with non-negativity constants (see Section \ref{sec_optimization}).
Additionally, we impose the
  cardinality constraints to  control the maximum number $C$ of non-zero values in $\bm{\theta}$ to mitigate overfitting (see Section \ref{sec_optimization}).  The value of the cardinality constraint is $C=1,2$.
As a comparison, we also use LASSO regression to select the quantile functions. The coefficient $\lambda$ in LASSO regression is determined as follows: $\lambda$ is increased with step $0.1$ until the cardinality (number of non-zero variables) equals $2$ and $1$. The smallest $\lambda$ that achieves the cardinality is reported.  We do not use the optimal weights formula \eqref{opt_weight} in this case study, because it assigns too small weights to extreme observations. Equal weights are assigned to all observations, as we are concerned about the tail risk.

\paragraph{Goodness-of-fit}
Table \ref{table_fit} presents various metrics to examine  goodness-of-fit. 
When the cardinality of the model is at least 2,  regressions outperform MLE in terms of WMSE and MAE. The error of the optimization with cardinality constraint 2 is only slightly higher compared to the optimization without cardinality constraint. The test statistics are significantly worse with cardinality 1. MLE outperforms other models in terms of KS and LLK. In most cases, the model fitted with cardinality constraint outperforms the model fitted with $L_1$ penalty. Visualization of the results are deferred to Appendix \ref{sec_qqplot}. 

\begin{table}[!htbp]
  \centering
  \renewcommand{\arraystretch}{1.2}
  \begin{tabular}{|p{1.1cm}|c|c|c|c|c|c|c|c|c|c|c|}
    \hline
    \multirow{2}{4cm}{} & \multicolumn{5}{c|}{\textbf{Least squares regression}} & \multicolumn{5}{c|}{\textbf{Least absolute deviation regression}}  & \textbf{MLE}\\
 
    \cline{1-12}
     $\bm{C}$ &  N/A & $\textbf{2}$ & $\textbf{1}$ & N/A &N/A & N/A & $\textbf{2}$ & $\textbf{1}$ & N/A & N/A & N/A \\
    \cline{1-12}
     $\bm{\lambda}$ & N/A & N/A & N/A & \textbf{0.6} & \textbf{1.2} & N/A & N/A & N/A & \textbf{1.1} & \textbf{1.9} & N/A \\
    \hline
    WMSE & $0.21$ & $0.21$ & $1.04$ & $0.25$ & $1.88$ & $0.31$ & $0.31$ & $0.28$ & $0.70$ & $3.07$ & $0.40$ \\ 
MAE & $0.14$ & $0.16$ & $0.34$ & $0.24$ & $0.78$ & $0.10$ & $0.10$ & $0.26$ & $0.38$ & $0.74$ & $0.15$ \\ 
KS & $0.10$ & $0.13$ & $0.20$ & $0.22$ & $0.57$ & $0.05$ & $0.08$ & $0.15$ & $0.18$ & $0.56$ & $0.03$ \\ 
LLK & $$-$381.0$ & $$-$386.3$ & $$-$439.4$ & $$-$585.2$ & $$-$737.4$ & $$-$389.2$ & $$-$400.6$ & $$-$414.6$ & $$-$518.5$ & $$-$720.5$ & $ $-$345.4$ \\    \hline
  \end{tabular}
  \caption{The table presents the measures of goodness-of-fit for models with different errors, constraints and penalties. MLE is included for comparison. $C = $ value of cardinality constraint, $\lambda =$ coefficient of $L_1$ penalty in LASSO regression. $\lambda$ is increased with step $0.1$ until the cardinality (number of non-zero variables) equals $2$ and $1$. 
  LLK = log likelihood.}
  \label{table_fit}
\end{table}

\section{Conclusion}\label{sec_conclusion}

We propose a flexible family of quantile function models estimated via constrained linear regression, leading to a convex optimization problem that is computationally efficient. The resulting estimator is shown to be asymptotically equivalent to a minimum  $q$-Wasserstein estimator and to be asymptotically normal under the stated regularity conditions. Practical regularization and constraints, including nonnegativity constraints, $L_1$, 
 penalties, cardinality constraints, and P-spline smoothing, can be conveniently incorporated to enhance finite-sample performance. In numerical experiments, the proposed model achieves superior fit in both the body and tail of the quantile function while exhibiting substantially faster computation than benchmark methods. The framework naturally extends to other distributional functionals that emphasize tail risk, such as CVaR and expectiles, which we leave for future research.

\section*{Statements and Declarations}

\textbf{Funding:}
The authors did not receive support from any organization for the submitted work.

\textbf{Conflict of Interest:}
The authors have no relevant financial or non-financial interests to disclose.

\clearpage

\bibliographystyle{chicago}
\bibliography{mixture_quantiles_reference}{}

\appendix

\section{Proof of Proposition \ref{prop_consistent}}\label{sec_proof_prop_consistent}
\begin{proof}
 $\mathcal{E}$ is a convex function with respect to $\bm{\theta}$. $\forall \bm{\theta} \in \bm{\Theta}$,  $f_N(\bm{\theta}) \overset{p}{\rightarrow} f(\bm{\theta})$ due to \eqref{converge_wass}. 
The rest is  proved by applying Theorem 2.7 of \cite{Newey2}. 
If the model is correctly specified,  we have $\widehat{\bm{\theta}}_N \overset{p}{\rightarrow} \widehat{\bm{\theta}}^*$. 
\end{proof}

\section{Proof of Proposition \ref{prop_integral_wass}}\label{sec_proof_prop_integral_wass}
\begin{proof}
We can view the  function $f_N(\bm{\theta})$ in the limit as a Monte Carlo integration. The convergence is derived by standard argument \citep{robert1999montecarlo}
\begin{equation}\label{converge_wass}
 \left( \frac{1}{N} \sum_{n=1}^N w_n \mathcal{E}_q(y_n - \bm{x}_n \bm{\theta}) \right)^{\frac{1}{q}} 
\overset{p}{\rightarrow}
 \left( \int_0^1 \mathcal{E}_q(Q(p) - G(p,\bm{\theta})) w(p) \mathrm{d}p \right)^{\frac{1}{q}} \;.
\end{equation}

\end{proof}

\section{Proof of Proposition \ref{prop_wass}}\label{sec_proof_prop_wass}
\begin{proof}
Since the model is correctly specified, we have $Q(p) = \sum_{i=0}^I \theta_{i}^* Q_i(p)$. Thus $Q(p)-G(\widehat{\bm{\theta}},p)
=
 \sum_{i=0}^I  (\theta_{i}^* - \widehat{\theta}_{Ni}) Q_i(p)
 $. 
We have
\begin{equation}
\begin{aligned}
\mathcal{W}(Q(p),G(\widehat{\bm{\theta}},p))) = &
\left( \int_0^1 \mathcal{E}_q(Q(p) - G(\widehat{\bm{\theta}},p)) w(p) \mathrm{d}p  \right)^{\frac{1}{q}} \\
 = &
 \left(  \int_0^1 \mathcal{E}_q\left( \sum_{i=0}^I  (\theta_{i}^* - \widehat{\theta}_{Ni}) Q_i(p)  \right) w(p) \mathrm{d}p  \right)^{\frac{1}{q}}
  \\ \leq &
  \sum_{i=0}^I |\theta_{i}^* - \widehat{\theta}_{Ni}| \left(  \int_0^1 \mathcal{E}_q\left(    Q_i(p)  \right) w(p) \mathrm{d}p \right)^{\frac{1}{q}}
  \\ \leq &
  M ||\bm{\theta}^{*} - \widehat{\bm{\theta}}_N||_1 \;.
\end{aligned}
\end{equation}
The inequality is due to subadditivity of  $\mathcal{E}$. 
By Proposition \ref{prop_consistent}, we have that $\mathcal{W}(Q(p),G(\widehat{\bm{\theta}},p)))  \overset{p}{\rightarrow} 0$. 
\end{proof}

\section{Proof of Proposition \ref{prop_orderstat}}\label{sec_proof_prop_orderstat}

\begin{proof} 
The difference between the sample quantiles and true quantiles asymptotically follows the normal distribution \citep{arnold2008order}, 
\begin{equation}\label{asymp_order}
\sqrt{N}
\left(
\bm{Y}_N - \bm{Y}^* 
\right)
\overset{d}{\rightarrow} 
\mathcal{N}(\bm{0},\bm{C}) \;.
\end{equation}
The solution to the estimation problem \eqref{statement} has a closed-form expression
\begin{equation} \label{closedform_inproof}
\widehat{\bm{\theta}}_N  = (\bm{X}'\bm{W}\bm{X})^{-1}\bm{X}'\bm{W}\bm{Y}_N \;.
\end{equation}
 Let $\bm{\epsilon}_N = \bm{Y}_N - \widehat{\bm{Y}}$. 
  Substituting $\bm{Y}$ with $\bm{X}\bm{\theta}^* + \bm{\epsilon}$ leads to
\begin{equation}
\widehat{\bm{\theta}}_N
= 
 (\bm{X}'\bm{W}\bm{X})^{-1}\bm{X}'\bm{W}\bm{X}\bm{\theta}^{*}
 +
 (\bm{X}'\bm{W}\bm{X})^{-1} \bm{X}'\bm{W} \bm{\epsilon}_N \;.
 \end{equation} 
 Thus 
 \begin{equation}
 \sqrt{N}
(
\widehat{\bm{\theta}}_N - \bm{\theta}^*
) = 
(\bm{X}'\bm{W}\bm{X})^{-1} \bm{X}'\bm{W}  
 \sqrt{N} \bm{\epsilon}_N \;.
 \end{equation}
With \eqref{asymp_order}, we have
\begin{equation}
\sqrt{N}
(
\widehat{\bm{\theta}}_N - \bm{\theta}^*
)
\overset{d}{\rightarrow} 
\mathcal{N}(\bm{0},\bm{H}) \;,
\end{equation}
where
\begin{equation}\label{asymp_cov}
\bm{H}
=
(\bm{X}'\bm{W}\bm{X})^{-1} \bm{X}'\bm{W} \bm{C} \bm{W}' \bm{X} (\bm{X}'\bm{W}'\bm{X})^{-1}  \;.
\end{equation}
By the generalized Gauss-Markov Theorem \citep{GeneralizedGMT1936Aitken}, the best linear unbiased estimator is obtained by the following weight matrix
\begin{equation}\label{optimal_weight}
\bm{W} = \bm{C}^{-1} \;.
\end{equation}
Substituting $\bm{W}$ with $ \bm{C}^{-1}$ in \eqref{asymp_cov} gives the optimal estimator in \eqref{opt_estimator}.
\end{proof}

\section{Constraints, Penalties and Error Function in Optimization} \label{sec_optimization}
This section contains further discussion on useful constraints and regularizations for the regression problem.

\subsection{Constraints}

This subsection lists some important constraints to be used in the regression.

\paragraph{Cardinality Constraint}
The cardinality constraint controls the maximum number $C$ of non-zero values in $\bm{\theta}$ to mitigate overfitting. The constraint is defined with the cardinality function,
\begin{equation}
Card(\bm{\theta})  = \sum_{i=1}^I \mathbf{1}_{\{\theta_i \neq 0\}} \leq C \;,
\end{equation}
where $\mathbf{1}_{\{\cdot\}}$ is the indicator function that equals 0 when the statement in the bracket is true and 0, if false.
The cardinality constraint can be transformed to a mixed integer linear system of constraints by introducing auxiliary integer variables $\bm{v}=(v_1,\cdots,v_I)$
\begin{equation}
\begin{aligned}
\sum_{i=1}^I v_i &\leq C \;, \\
l_i v_i &\leq \theta_i \leq u_i v_i \;, \\
v_i &\in \{ 0,1 \} \;,
\end{aligned}
\end{equation}
where $u_i$ and $l_i$ are upper and lower bounds for $\theta_i$.
 Mixed integer linear programming can be effectively solved using powerful solvers, such as GUROBI. The cardinality constraint allows for including a large number of quantile functions in the initial mixture.

\paragraph{VaR and CVaR Constraints}

Since we assign small weights to tail observations, the fitted distribution could underestimate the heaviness of tail. By imposing constraints on the tail behavior, we can improve the representation of the distribution in the tails. In some practical applications tail performance is well studied (for instance, in finance applications).
The quantile is also called Value-at-Risk (VaR) in finance. 

 CVaR is a popular risk measure  in finance and engineering that accounts for the average of the tail of a distribution \citep{cvar1,cvar2}. 
The $p$-CVaR of a random variable with quantile function $Q(p)$ is defined by $\bar{Q}(p) = \frac{1}{1-p}\int_p^1 Q(p)\mathrm{d}p$. Equivalently, CVaR for a continuous distribution is the average outcome exceeding the $p$-quantile and thus contains more information on the tail than  $p$-quantile. From the definition we observe that the $p$-CVaR of $G(p,\bm{\theta})$ is  linear with respect to parameters $\bm{\theta}$.

The $p$-quantile constraint is defined by
\begin{equation}
\sum_{i=0}^I \theta_i Q_i(p) \leq E \;,
\end{equation}
and the $p$-CVaR constraint is defined by
 \begin{equation}
\sum_{i=0}^I \theta_i \bar{Q}_i(p) \leq F \;,
\end{equation}
where $E, F$ are constants. The direction of the inequality sign indicates that we do not want to underestimate the upper tail risk, but similar we can control the lower tail as well.

\paragraph{L-moment Constraints} 
The L-moments are defined by a linear combination of order statistics. These moments can characterize a distribution \citep{Hosking}, similar to the moments of a distribution, but are more robust to outliers.
With a minor abuse of notation, denote by $L_m(Q)$ the  $m$-th L-moment of a sample set $Q$  of a random variable whose quantile function is $Q$, and by $\bm{L}_m$  the vector of L-moments $(L_1,\cdots,L_m)'$.

\cite{Hosking} gives the analytical expression of L-moment $L_m(Q) = \int_0^1 Q(p)F_m(p)\mathrm{d}p$ where $F_m(p)$ is a polynomial. We observe that the L-moment of a sum of quantile function of $G(p,\bm{\theta})$ is linear with respect to $\bm{\theta}$
\begin{equation}
L_m(G(\bm{\theta})) = L_m( \sum_{i=0}^I  \theta_i Q_i ) = \sum_{i=0}^I  \theta_i L_m( Q_i ) 
\end{equation}
where $L_m( Q_0 )=1$.

L-moment fitting in \cite{Karvanen}, analogous to method of moments, matches the L-moments of data and the model.
For our estimation problem, we propose a constraint that bounds the $L_1$ error of L-moments of fitted distribution from the empirical $L$-moments as follows
\begin{equation}
	\label{27}
||(\bm{L}_m ( G(\bm{\theta})) - \bm{L}_m(\bm{Y})||_1 \leq \epsilon \;.
\end{equation}
L-moment fitting is a special case of this procedure, since the feasible set shrinks to the optimal solution of L-moment fitting as $\epsilon \rightarrow 0$.
The constraint 	\eqref{27} is equivalent to the following system of linear constraints with auxiliary variables $ u_1,\cdots,u_m$ ,
\begin{equation}
\begin{aligned}
\sum_{k=1}^m u_i &\leq \epsilon  \;,  \\
u_k & \geq \sum_{i=0}^I  \theta_i L_k( Q_i )  - {L}_k(\bm{Y})  , \quad k = 1, \cdots,m \;,\\
u_k & \geq -\sum_{i=0}^I  \theta_i L_k( Q_i )  + {L}_k(\bm{Y})  , \quad k = 1, \cdots,m \;.
\end{aligned}
\end{equation}


\subsection{Error Function and Penalty}

While this paper mainly considers weighted mean squared error and mean absolute error, other error functions for $\bm{\mathcal{E}}$ can be chosen, as well. For example, the robust Huber error \citep{huber1964robust}  combining mean squared  and mean absolute errors can be useful for handling outliers.

\section{Q-Q Plots of Experiment in Section \ref{sec_dd}}\label{sec_qqplot}

Figure \ref{qqplot_ls} and \ref{qqplot_lad} present the Q-Q plot of models fitted with different errors, constraints and penalties. MLE is included for comparison. 
The following three regressions provide close results and have good fit: 1) weighted least squares regression; 2) weighted least squares regression  with cardinality constraint $C=2$; 3) weighted least squares regression with Lasso penalty coefficient $\lambda = 0.6$.  These models have a better fit to the tail observations than MLE. Figures \ref{plot_qq_LS_MLE},   \ref{plot_qq_LAD_MLE}  for MLE show that extreme tail observations are all  below the straight line $y=x$, i.e., MLE underestimates the tail risk.  When cardinality equals $1$, neither of the regression with cardinality constraint or  with $L_1$ penalty has a good fit. 
For least absolute deviation regression, 
except for the regression with coefficient of $L_1$ penalty $\lambda = 1.9$, all models  are close and  have good fit. They have better fit to the tail observations than MLE. 

\begin{figure}[htbp] 
\centering
  	\subfloat[ Weighted least squares]{%
    \includegraphics[width=.32\textwidth]{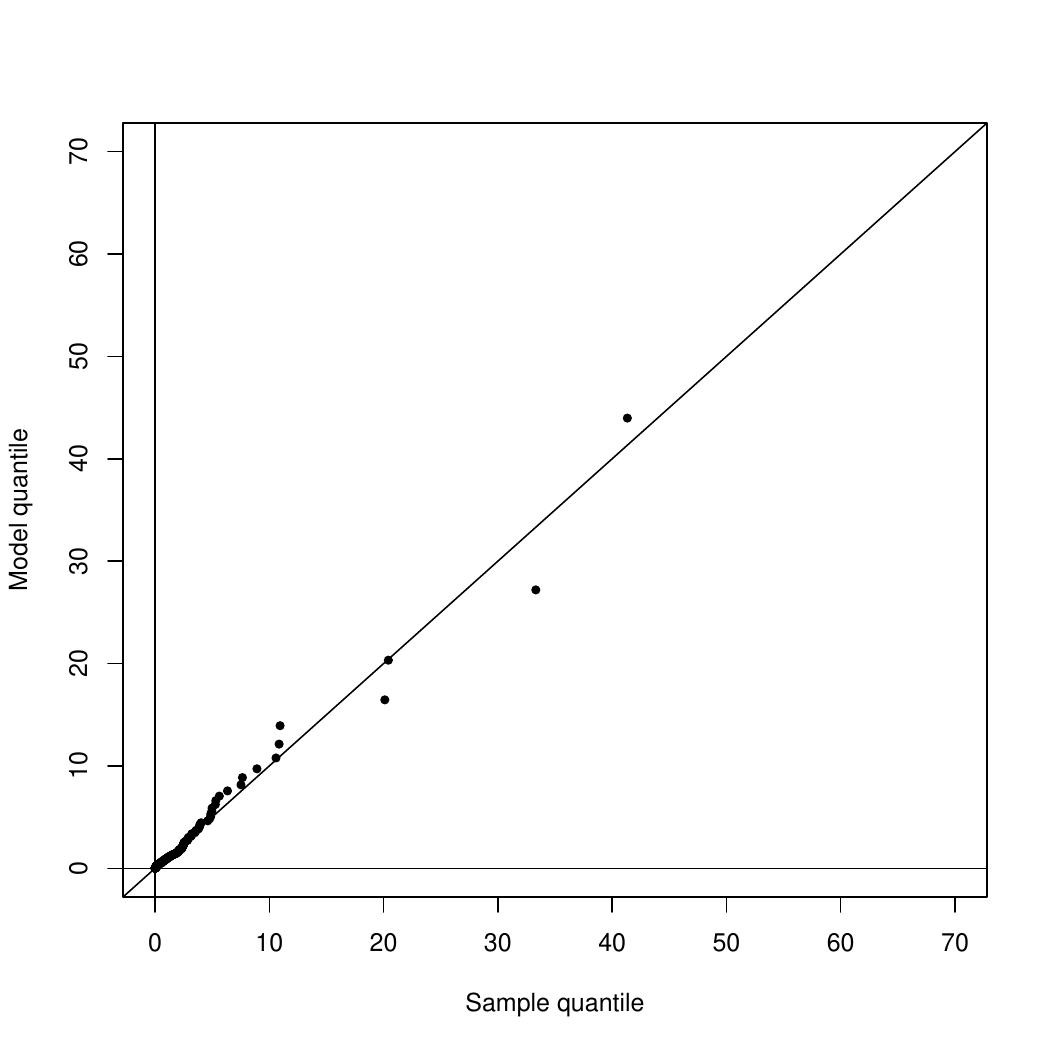}}
   	\hspace{0.12cm}
    \subfloat[ MLE]{%
    \includegraphics[width=.32\textwidth]{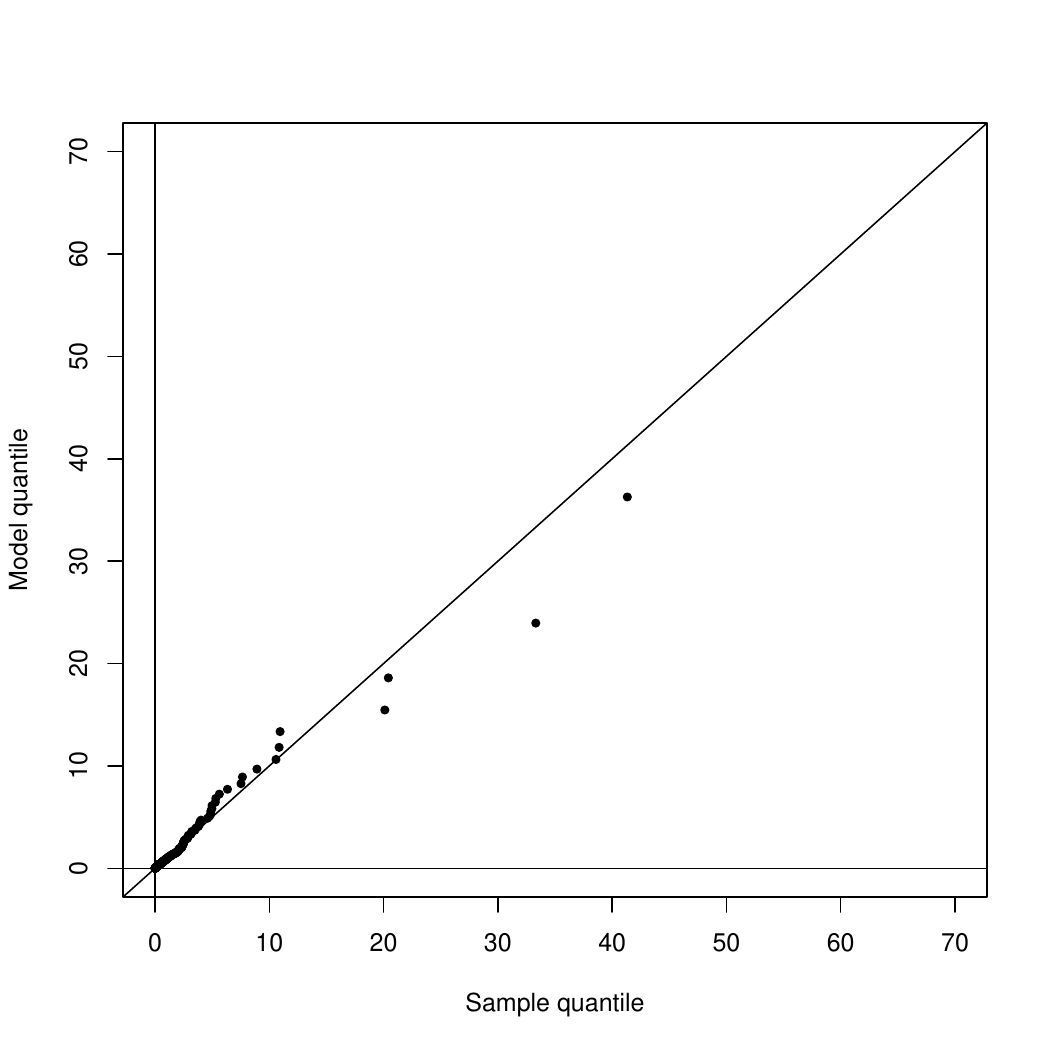}\label{plot_qq_LS_MLE}  }
    \hspace{0.12cm}
    \subfloat[ Cardinality constraint $C=2$]{%
    \includegraphics[width=.32\textwidth]{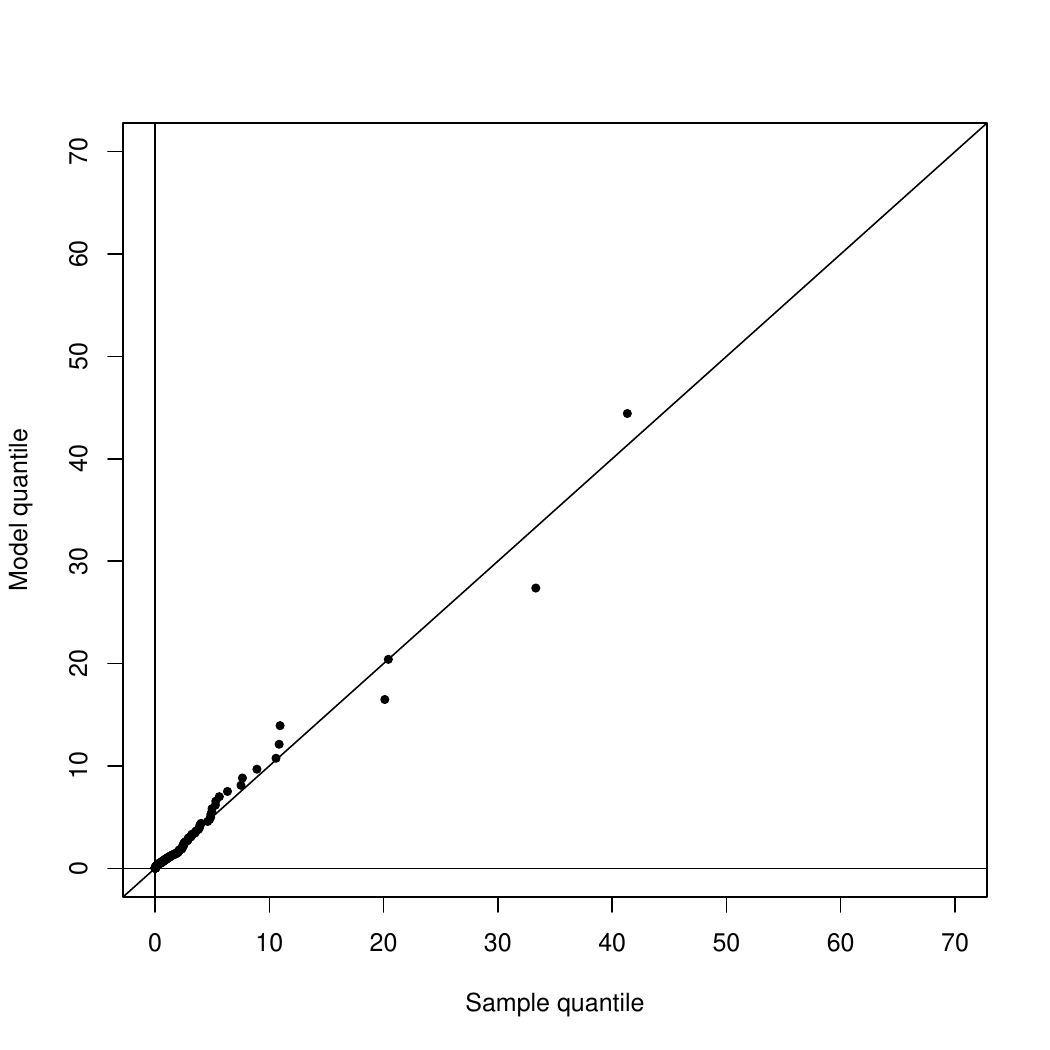}}
    \\
    \hspace{0.12cm}
    \subfloat[ Cardinality constraint $C=1$]{%
    \includegraphics[width=.32\textwidth]{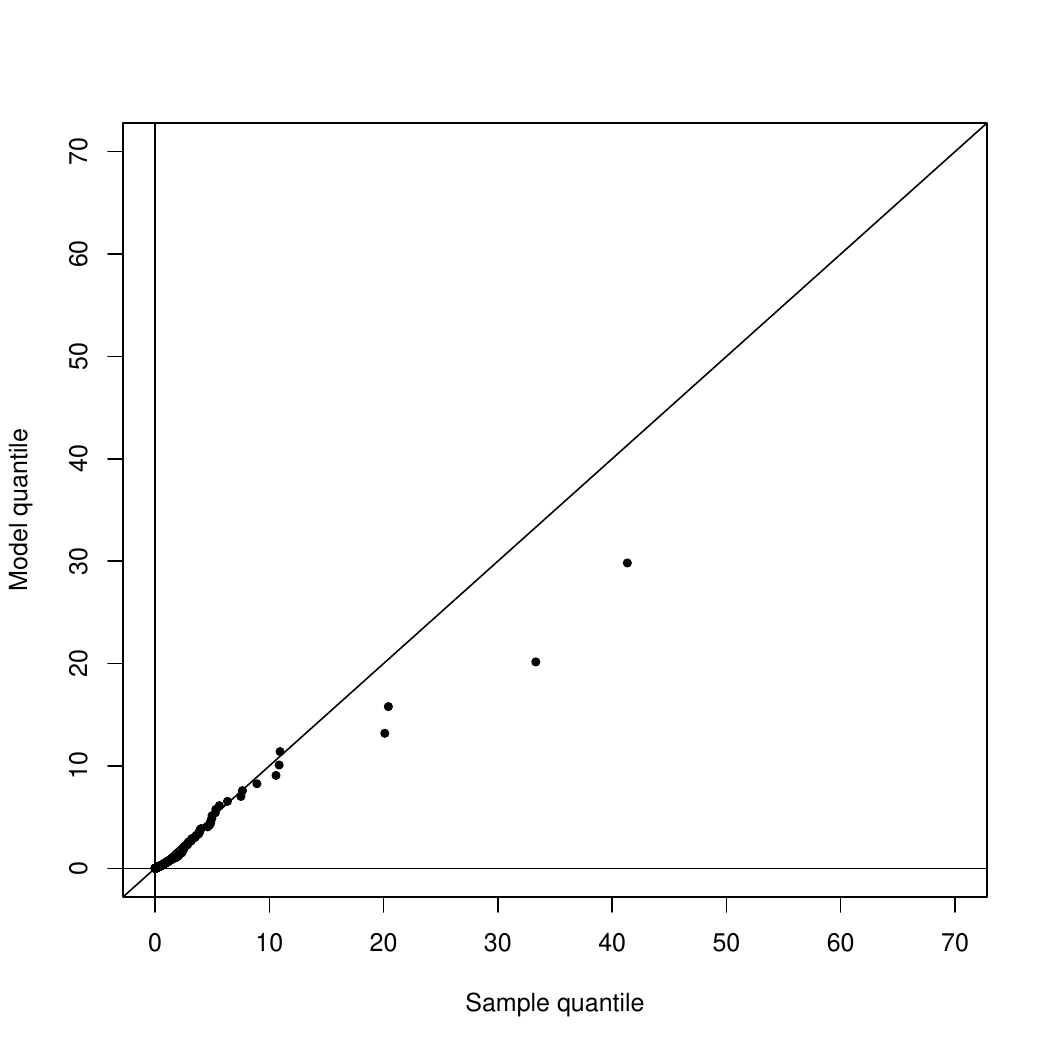}}
    \hspace{0.12cm}
    \subfloat[ LASSO $\lambda = 0.6$]{%
    \includegraphics[width=.32\textwidth]{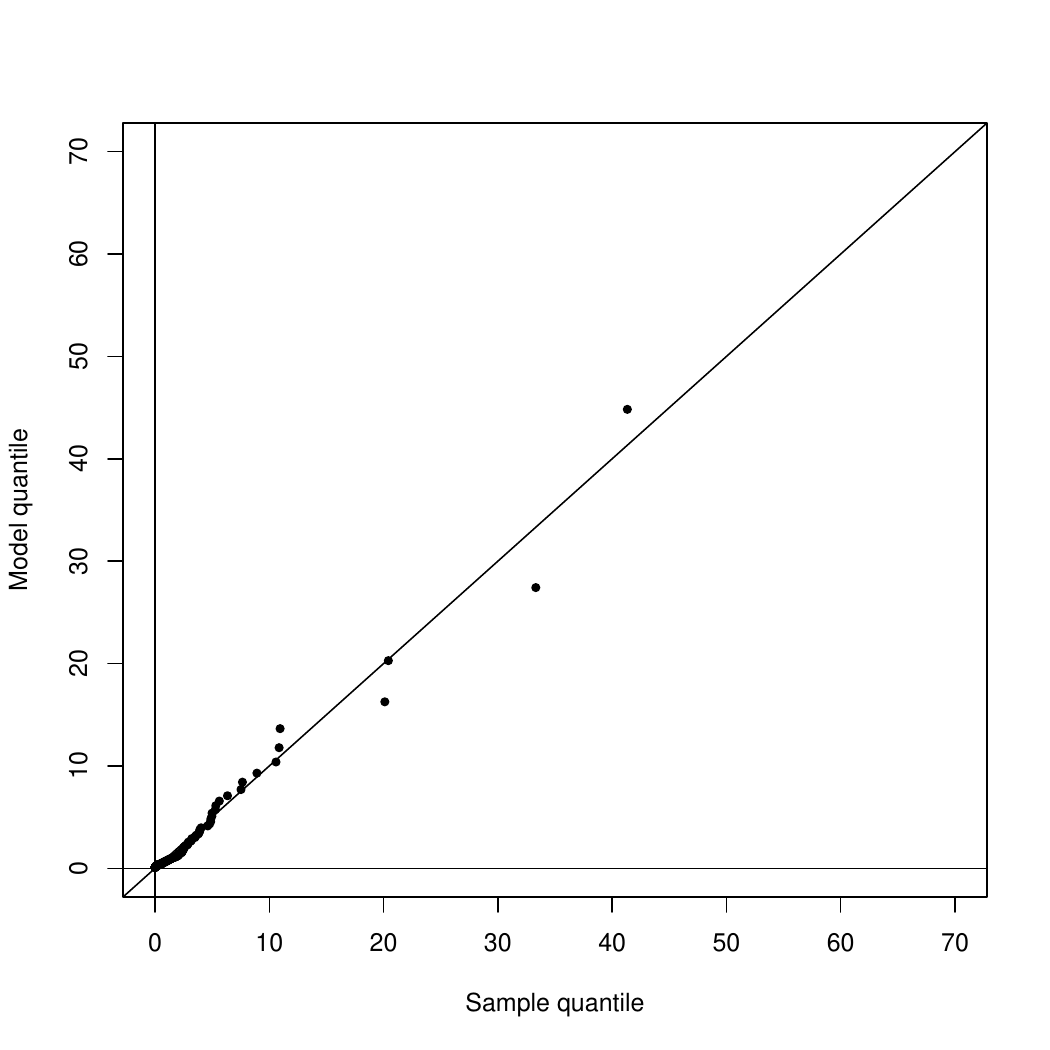}}
    \hspace{0.12cm}
    \subfloat[ LASSO $\lambda = 1.2$]{%
    \includegraphics[width=.32\textwidth]{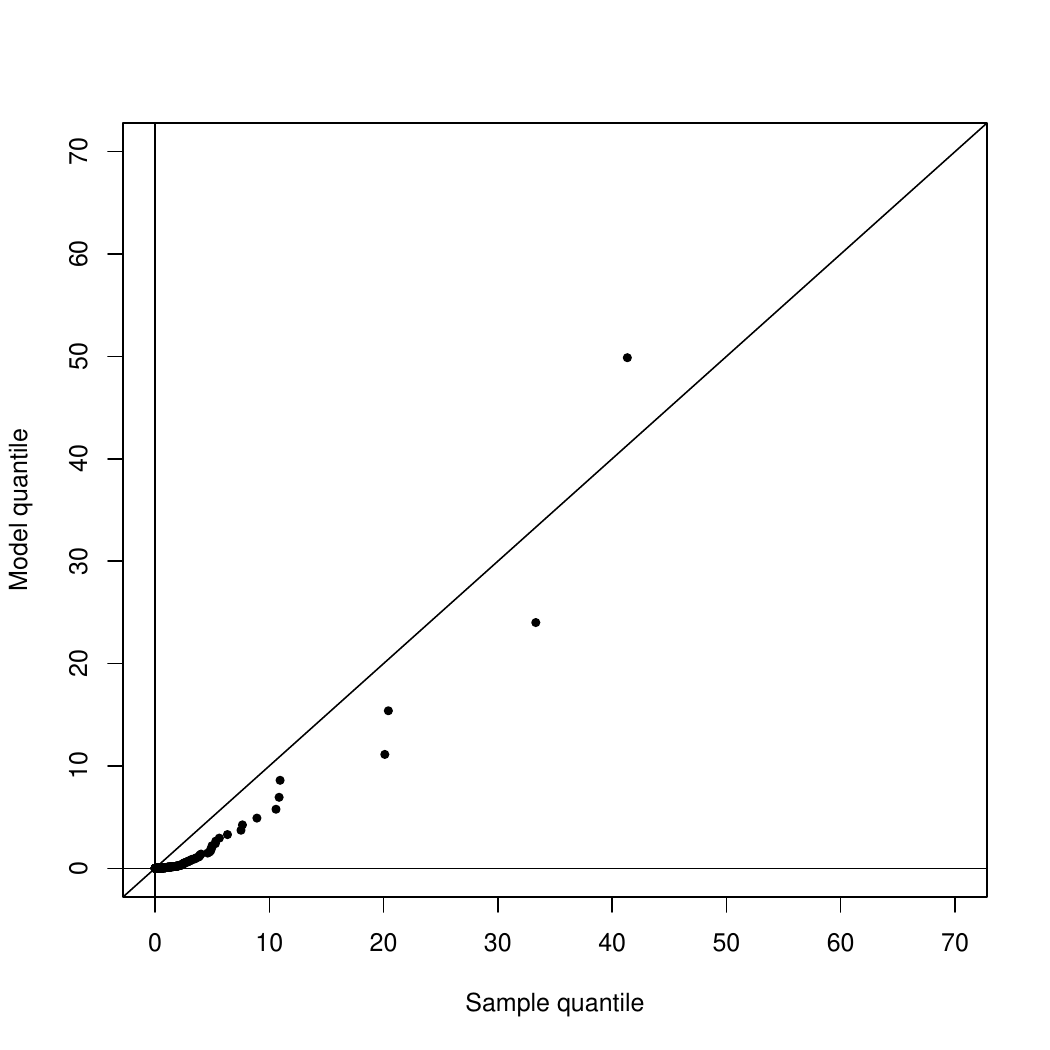}}
 
  \caption{Q-Q plots of models fitted by least squares regression with  cardinality constraint $C=1,2$ and coefficient $\lambda=0.6,12$ of $L_1$ penalty. MLE is included as a benchmark.  $\{(x_n,y_n)\}_{n=1}^N$ = 
 black points,   	
  $x_n=$ $n$-th sample order statistics, $y_n=$ quantile with  confidence level $\frac{n}{N+1}$ of the model. }
  \label{qqplot_ls}
\end{figure}

\begin{figure}[!htbp]
\centering
  	\subfloat[ Least absolute deviation]{%
    \includegraphics[width=.32\textwidth]{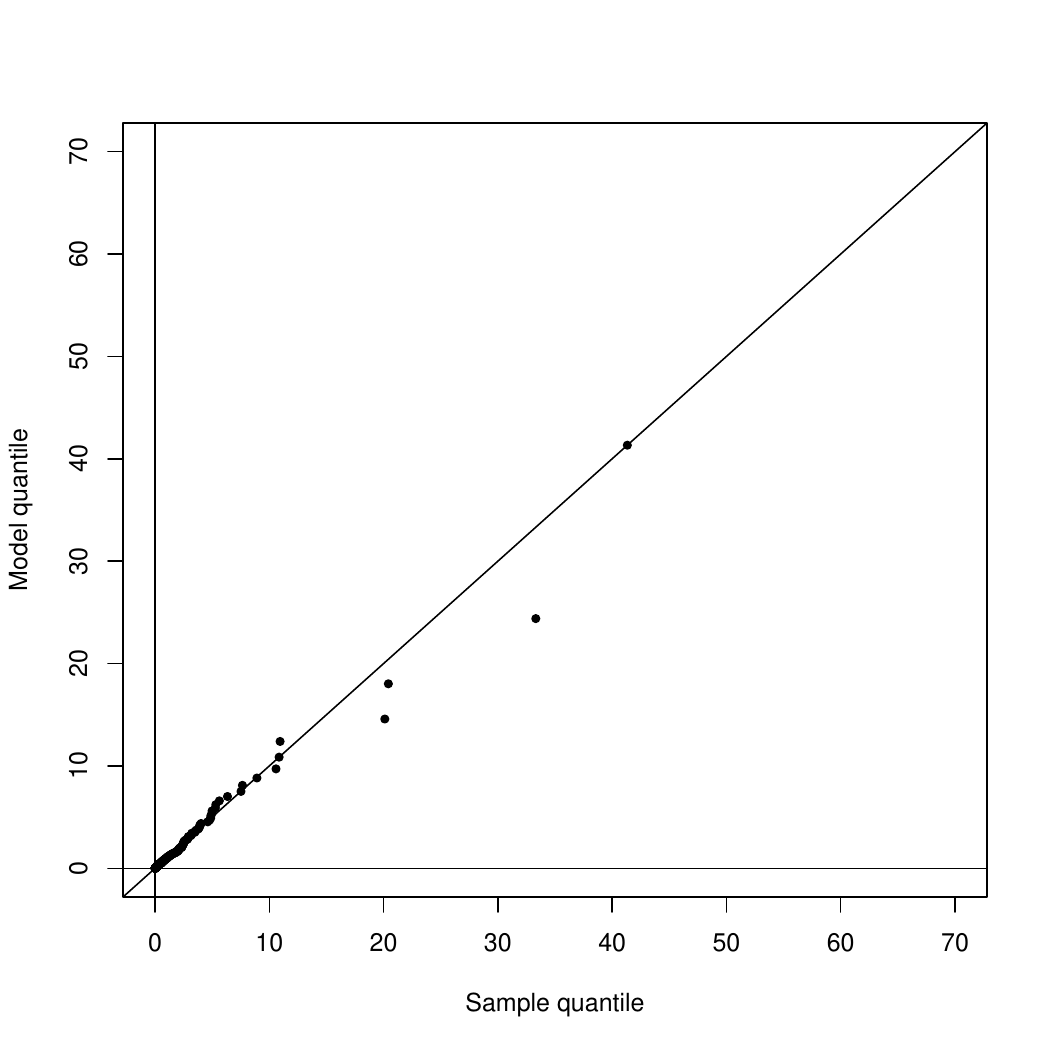}}
   	\hspace{0.12cm}
    \subfloat[ MLE]{%
    \includegraphics[width=.32\textwidth]{qq_plot_11.pdf}\label{plot_qq_LAD_MLE} }  
    \hspace{0.12cm}
    \subfloat[ Cardinality constraint $C=2$]{%
    \includegraphics[width=.32\textwidth]{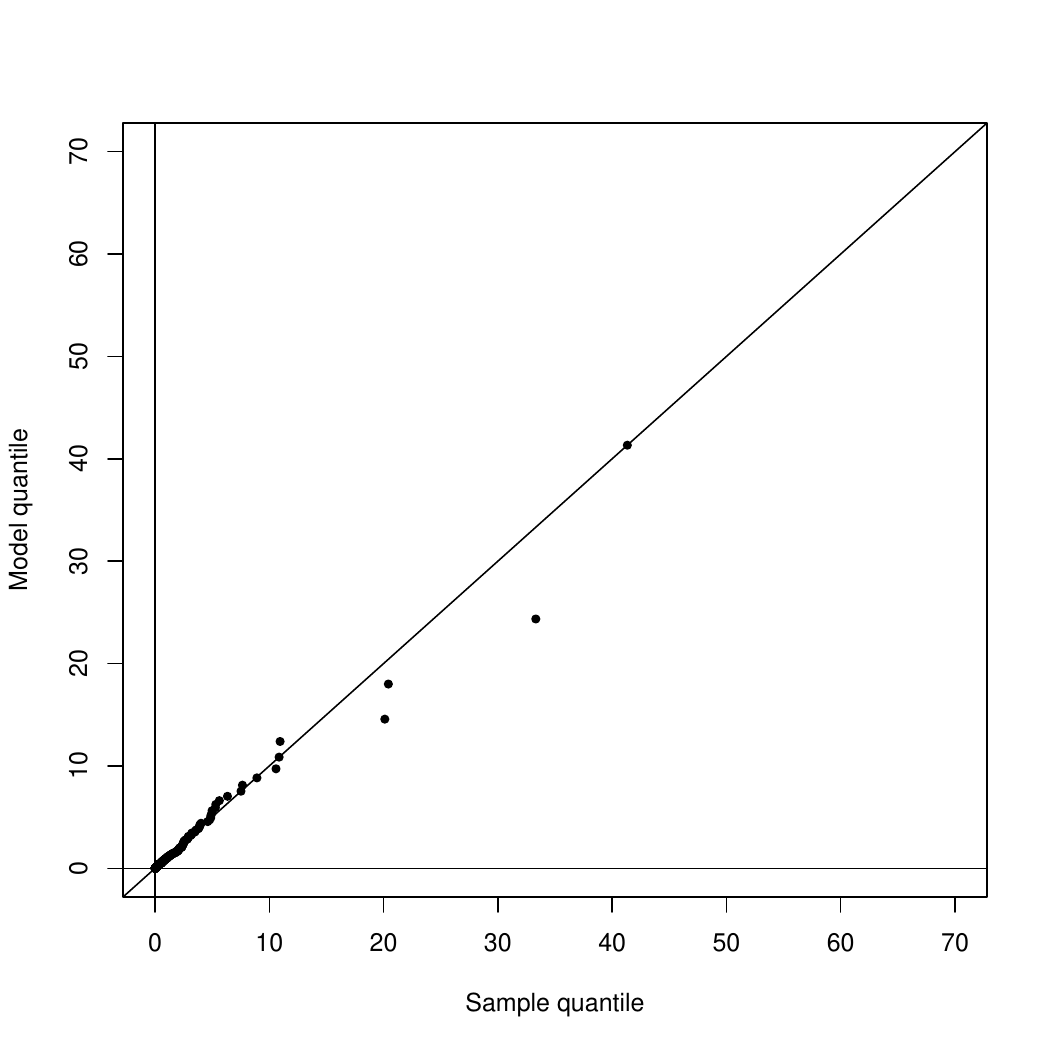}}
    \\
    \subfloat[ Cardinality constraint $C=1$]{%
    \includegraphics[width=.32\textwidth]{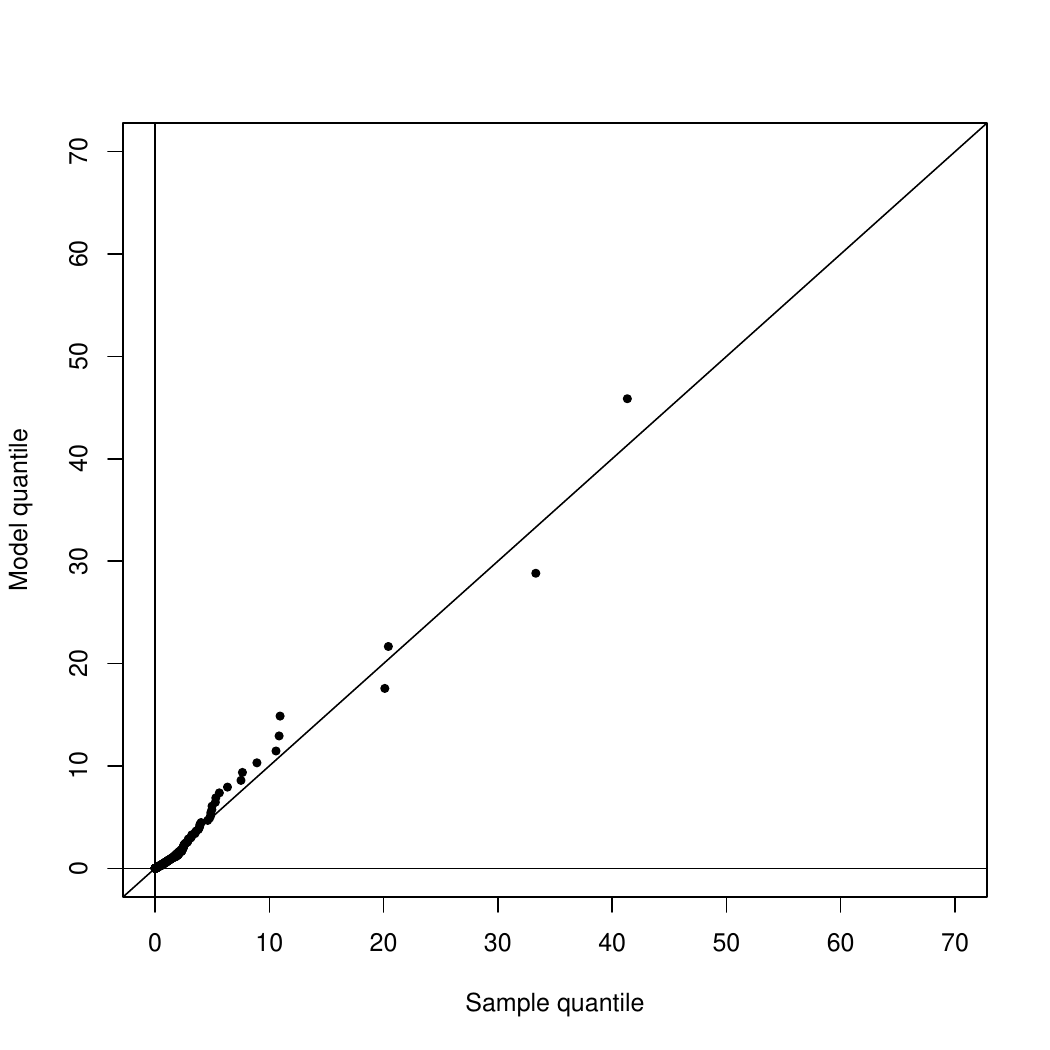}}
    \hspace{0.12cm}
    \subfloat[ LASSO $\lambda = 1.1$]{%
    \includegraphics[width=.32\textwidth]{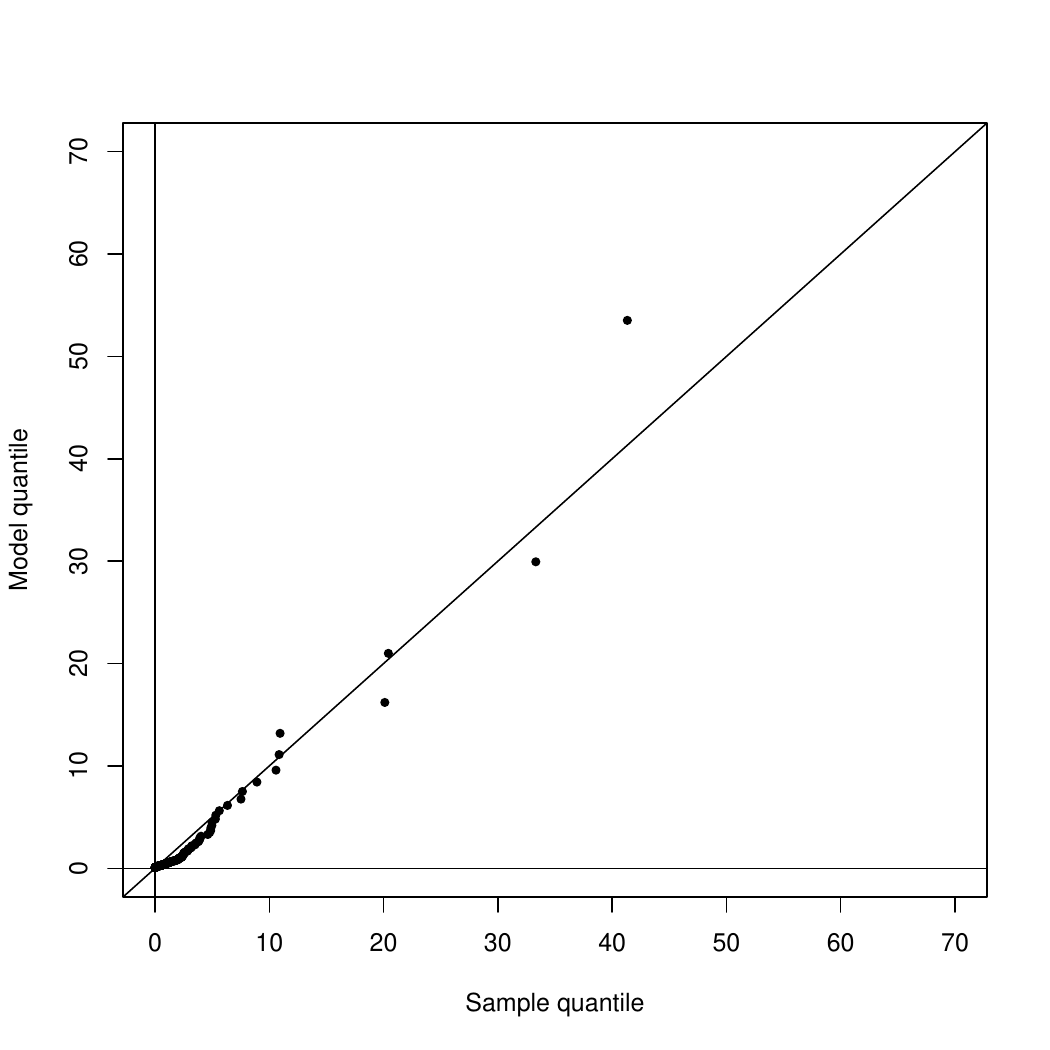}}
    \hspace{0.12cm}
    \subfloat[ LASSO $\lambda = 1.9$]{%
    \includegraphics[width=.32\textwidth]{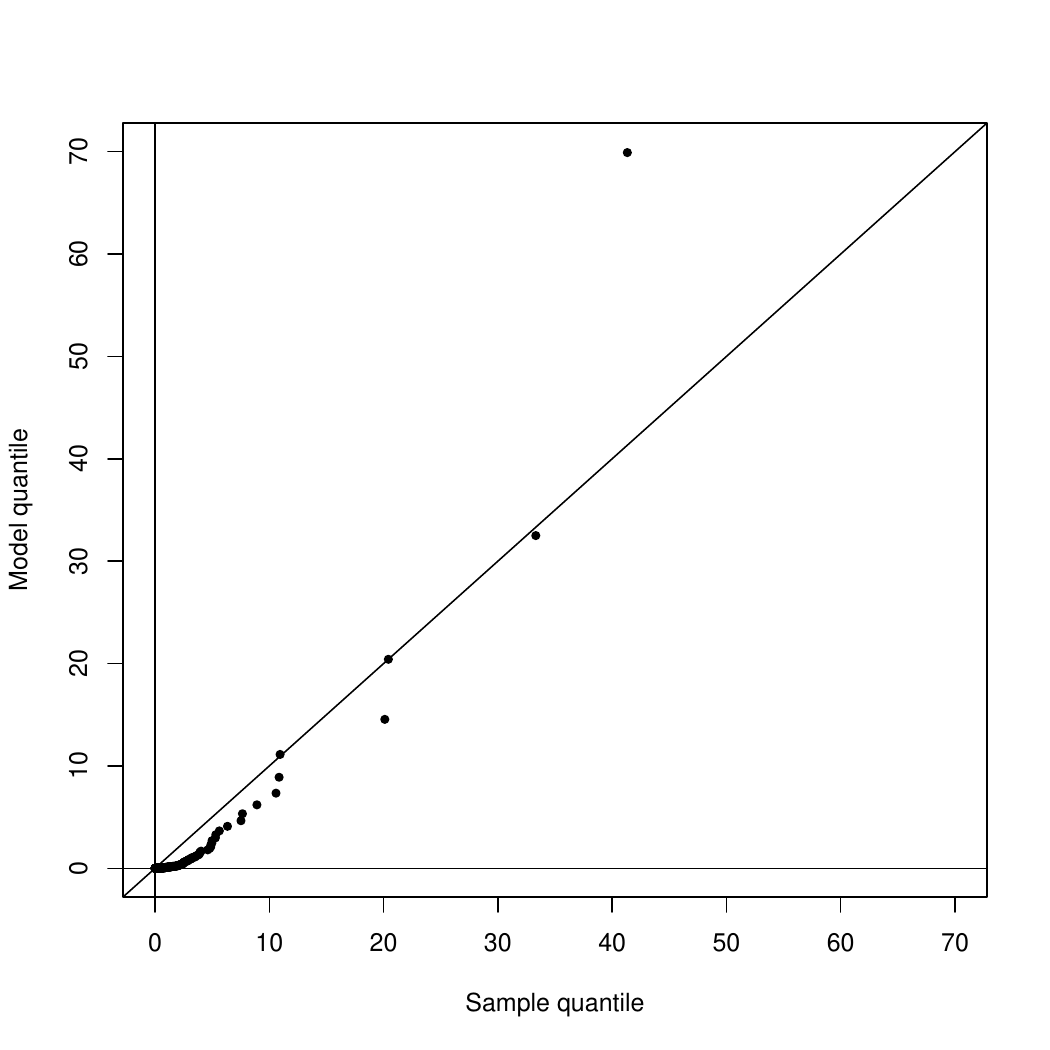}}
 
  \caption{Q-Q plots of models fitted by least absolute deviation regression with  cardinality constraint $C=1,2$ and coefficient $\lambda=1.1,1.9$ of $L_1$ penalty.   	  	
  MLE is included as a benchmark.  $\{(x_n,y_n)\}_{n=1}^N$ = 
  black points,   	
  $x_n=$ $n$-th sample order statistics, $y_n=$ quantile with  confidence level $\frac{n}{N+1}$ of the model.}
  \label{qqplot_lad}
\end{figure}

\section{Simulation and Analysis of  Convergence}\label{sec_sim_convergence}

This section studies the convergence of Wasserstein distance bewteen the mixture quantile model and the true quantile function with a simulated dataset.
We fit the model to the sample data-set generated by a mixture of skewed $t$ distributions.
We calculated distance between the quantile functions of the underlying data-generating model and of the fitted  model. Section \ref{sec_minimum_distance} shows that the objective value of \eqref{statement} converges to the $p$-Wasserstein distance. 
We show with graphs that as the sample size $N$ increases, the distance between the model quantile function and the true quantile function tends to zero, and that the objective value tends to zero. 

\cite{Fernandez} modifies skewed $t$ distribution by scaling the two sides of the density function differently. Similarly, we modify skewed $t$ distribution by scaling differently the positive and negative sides of the quantile function. The quantile function of the skewed $t$ distribution is defined by 
\begin{equation}\label{skewt}
 Q_{\gamma,\nu}^{st}(p) = 
 \begin{cases} 
        \frac{1}{\gamma R} Q_\nu^t(p)  \;,
      &  p  \leq 0.5  \\
        \frac{\gamma}{R} Q_\nu^t(p)  \;,
       &  p > 0.5 \\
   \end{cases} \;,
\end{equation}
where $p$, $\gamma$, $\nu$ are the probability, skewness parameter and degrees of freedom, $Q_\nu^t$ is the quantile function of standard $t$ distribution with degrees of freedom $\nu$, $R$ is an normalizer to obtain unit quartile range. The median of the skewed $t$ distribution is zero.

Consider the model 
\begin{equation}
G(p,\bm{\theta}) = \sum_{i=1}^I \theta_i Q_i(p) \;,
\end{equation}
where $\{Q_i(p)\}_{i=1}^I
=
\{
Q^{st}_{\gamma,\nu}(p) |
\gamma \in \{ \frac{1}{2},\frac{1}{1.8},\frac{1}{1.6},\frac{1}{1.4},\frac{1}{1.2},1,1.2,1.4,\cdots,2 \}$, $\nu \in \{ 5,7,9,\cdots, 25\}
\}$. The true parameter $\bm{\theta}$ is set as $(1,1,1,\cdots,1,0.2,0.2,\cdots,0.2)$, i.e., besides intercept $1$, $Q_i(p)$ with $\gamma>1$ have coefficient $1$ while the rest have $0.2$. We fit the model to the $N$ samples generated from the true model with inverse transform sampling. We use the integer part of $(10^2,10^{2.25},10^{2.5},10^{2.75},\cdots,10^4)$ as the sample size $N$. For each sample size, the experiment is conducted for $100$ times to obtain the mean and confidence band.

  In this case study, the weights  of observations are the diagonal elements of the optimal weight matrix \eqref{opt_weight}. Since the density function in  \eqref{opt_weight} is not known apriori in practice, we plug in normal distribution to obtain the weights. The weights are  $w_{n} = \frac{f(Q(p_n))^2}{p_n(1-p_n)}$, where $p_n = \frac{n}{N+1}$, and $f,Q$ are the density  and quantile of a standard normal distribution.  Numerical experiments in \cite{Ergashev} show that the impact of the weights is insignificant  as long as the extreme tail observations are assigned with small weights.

 Figure \ref{convergence_ls} and Figure  \ref{convergence_lad} 
 show that the Wasserstein distance between the estimated quantile function and the true quantile function converges to zero as the sample size $N$ increases. 

\begin{figure}[htbp] 
\centering
\includegraphics[scale=0.65]{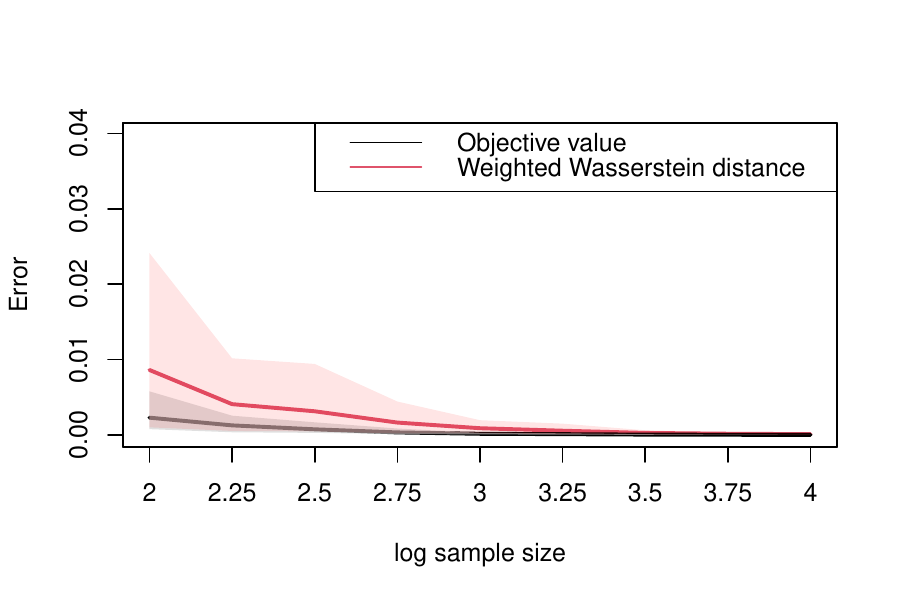}
\caption{
Convergence of error and weighted 2-Wasserstein distance obtained by weighted least squares regression. Lower (black) line
= objective (error) of the optimization problem of weighted least squares regression. Upper (red) line = Wasserstein distance between the
estimated quantile function and the true quantile function.  
Wide (grey) band = $90$\% confidence band of the error  obtained by $100$ repeated experiments. Thin (red) band = $90$\% confidence band of the distance obtained by $100$ repeated experiments. The horizontal axis = sample size in log scale.
}
\label{convergence_ls} 
\includegraphics[scale=0.65]{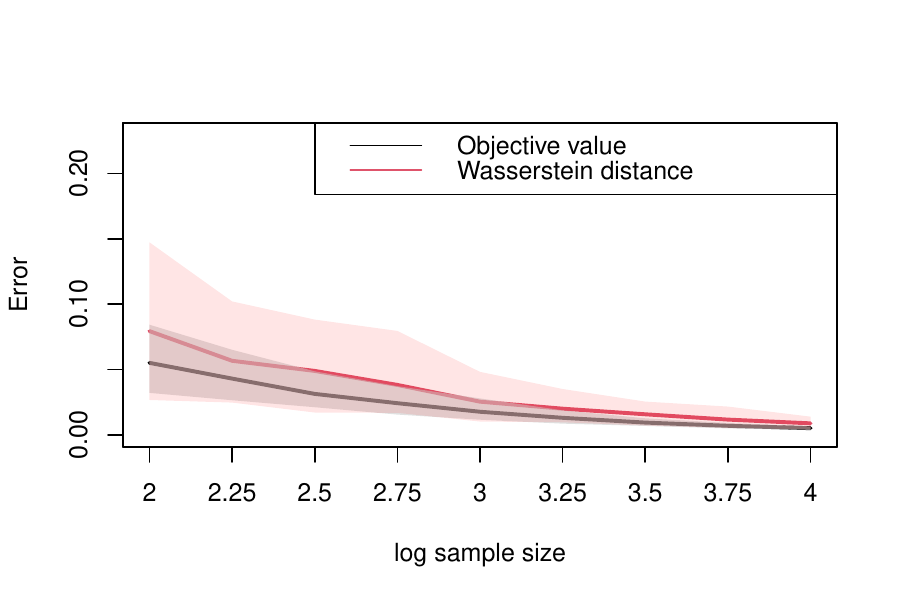}
\caption{
Convergence of error and 1-Wasserstein distance obtained by  least absolute deviation regression. Lower (black) line
= objective (error) of the optimization problem of least absolute deviation regression. Upper (red) line = Wasserstein distance between the
estimated quantile function and the true quantile function.  
Wide (grey) band = $90$\% confidence band of the error  obtained by $100$ repeated experiments. Thin (red) band = $90$\% confidence band of the distance obtained by $100$ repeated experiments. The horizontal axis = sample size in log scale.
}
\label{convergence_lad}
\end{figure}

\end{document}